\documentclass[final,3p,times]{elsarticle}

\usepackage{amssymb}   
\usepackage{caption}   
\usepackage{subfigure} 
\usepackage{amsmath}   
\usepackage{setspace}  
\usepackage{algorithm} 
\usepackage{algpseudocode} 
\usepackage{makecell}  
\usepackage{multirow}  
\usepackage[justification=centering]{caption} 
\usepackage{multicol}
\usepackage{stfloats}  
\usepackage{color}     
\usepackage{amsthm}    
\usepackage{threeparttable}
\usepackage{url}

\usepackage{dcolumn}

\usepackage{algorithm}
\usepackage{algpseudocode}

\newtheorem{definition}{Definition}
\newtheorem{example}{Example}
\newtheorem{theorem}{Theorem}

\begin{document}

\begin{frontmatter}
\title{Top-\textit{k} contrast order-preserving pattern mining}
  
\author[1,3]{Youxi Wu}	

\author[1,2]{Yufei Meng}	
  
	\author[3]{Yan Li}
		
	\author[4]{Lei Guo}
		
	\author[5]{Xingquan Zhu}
	
	\author[6]{Philippe Fournier-Viger}
	
	\author[7]{Xindong Wu}
		
		
		\address[1]{School of Artificial Intelligence, Hebei University of Technology, Tianjin 300401, China}
		
		\address[2]{Shijiazhuang University of Applied Technology, Shijiazhuang 050081, China}
		
		\address[3]{School of Economics and Management, Hebei University of Technology, Tianjin 300401, China}
		
		\address[4]{State Key Laboratory of Reliability and Intelligence of Electrical Equipment, Hebei University of Technology, Tianjin 300401, China}
		
		\address[5]{Department of Computer \& Electrical Engineering and Computer Science, Florida Atlantic University, FL 33431, USA}
		
		\address[6]{Shenzhen University, Shenzhen, China}
		
		\address[7]{Research Center for Knowledge Engineering at the Zhejiang Lab, Hangzhou 311121, China}
		
\begin{abstract}
Recently, order-preserving pattern (OPP) mining, a new sequential pattern mining method, has been proposed to mine frequent relative orders in a time series. Although frequent relative orders can be used as features to classify a time series, the mined patterns do not reflect the differences between two classes of time series well. To effectively discover the differences between time series, this paper addresses the top-\textit{k} contrast OPP (COPP) mining and proposes a COPP-Miner algorithm to discover the top-\textit{k} contrast patterns as features for time series classification, avoiding the problem of improper parameter setting. COPP-Miner is composed of three parts: extreme point extraction to reduce the length of the original time series, forward mining, and reverse mining to discover COPPs. Forward mining contains three steps: group pattern fusion strategy to generate candidate patterns, the support rate calculation method to efficiently calculate the support of a pattern, and two pruning strategies to further prune candidate patterns. Reverse mining uses one pruning strategy to prune candidate patterns and consists of applying the same process as forward mining. Experimental results validate the efficiency of the proposed algorithm and show that top-\textit{k} COPPs can be used as features to obtain a better classification performance.

\end{abstract}
		

\clearpage

		\end{frontmatter}
		
		\section{Introduction}
		A time series is a series of observed numerical values indexed in time order. Time series exist in various aspects of life; for example, there are time series of the daily stock price \cite{li2021tkde} and data collected by various sensors \cite{2021tkde}, and they exist in ECG/EEG data \cite{Dai2022tcyb}. Time series contains a lot of useful information for users. Various time series mining methods have been proposed. For example, Wu and Keogh \cite{wukeogh}  focused on research on time series anomaly detection and established a UCR time series anomaly archive. To discover patterns composed of strong and medium fluctuations that users are usually more concerned about, three-way sequential pattern mining {\cite{wu2022ntpm,min2020freq}} divides the fluctuations of a time series into strong, medium, and weak fluctuations, and patterns are composed of strong and medium fluctuations. To handle streaming trajectories,  real-time co-movement pattern mining was proposed {\cite{fang2020coming,lu2019real}}.  To classify time series, Bai et al. {\cite{bai2021time}} proposed an effective method for extracting both the mean features and trend features of a time series. To avoid a large deviation between the query pattern and time series, (delta, gamma) approximate sequence pattern mining methods {\cite{li2021netd,wu2020netd}} were investigated since they can evaluate local and global distances at the same time. All of the above methods adopt some symbolic methods to convert time series into symbols, since a time series consists of numerical values that are difficult to handle. Thus, various symbolic methods, such as the piecewise aggregate approximation  method {\cite{keogh2001loca}} and symbolic aggregate approximation  {\cite{lin2007expe}} method, were explored. However, these symbolic methods not only introduce data noise but also pay too much attention to the values of the time series and ignore the trend changes of the time series. 
		
		Recently, the order-preserving pattern (OPP) matching {\cite{kim2014orde}} method has been proposed, and it uses a group of relative orders to represent a pattern, which can represent a trend in a time series. Inspired by OPP matching, OPP mining, a new sequential pattern mining method, was proposed to discover frequent OPPs in a given time series {\cite{wu2022oppm}}. Then, order-preserving rule (OPR) mining was studied to discover the OPRs with high confidence rate {\cite{wu2022oprm}}. Although the experimental results showed that both OPPs \cite{wu2022oppm} and OPRs \cite{wu2022oprm} can be used to extract features for time series classification, these methods only discover frequent patterns (common characters) as features but neglect the differences between different classes. For example, in Fig. \ref{database}, the OPP of the sub-time series (2,6) in $\textbf{t}_1$ is (1,2), since 2 is smaller than 6. Similarly, the OPPs of the sub-time series (2,4) in $\textbf{t}_2$, (2,5) in $\textbf{t}_3$, (4,7) in $\textbf{t}_4$, (6,9) in $\textbf{t}_5$, and (5,7) in $\textbf{t}_6$ are also (1,2). Thus, pattern (1,2) is a frequent OPP and occurs in each time series in $D_+$ and $D_-$, which cannot show the significant differences in different classes.
		Therefore, it is difficult to use frequent OPPs as classification features to achieve a good classification performance.
		
		\begin{figure}
			\centering
			\includegraphics[width=1\linewidth]{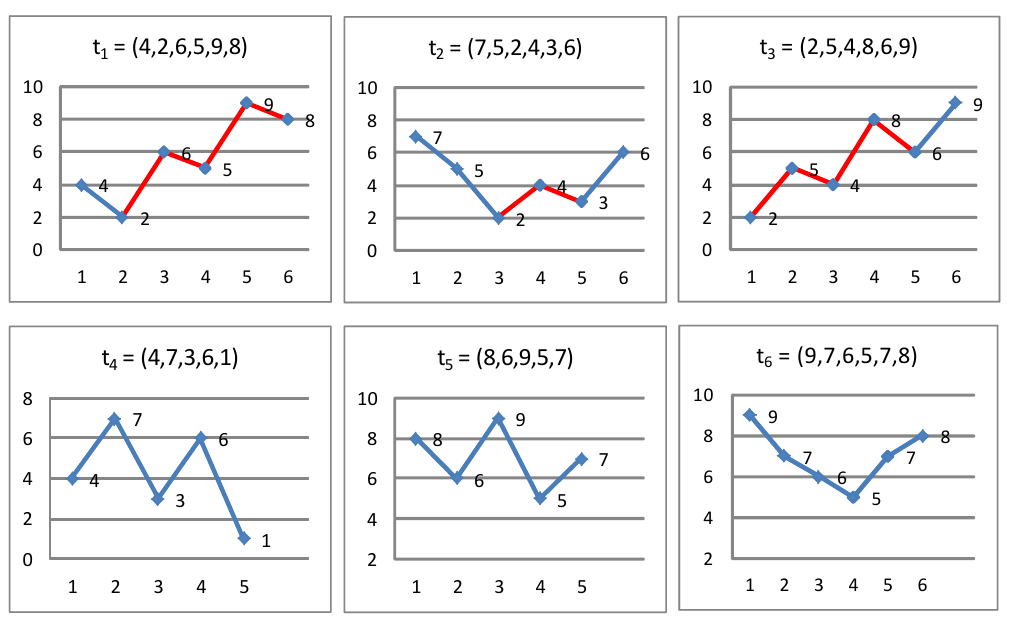}
			\captionsetup{font=footnotesize} 
			\caption{Six time series with two labels $D_+ = \{\textbf{t}_1, \textbf{t}_2, \textbf{t}_3\}$ and $D_- = \{\textbf{t}_4, \textbf{t}_5, \textbf{t}_6\}$. The OPP of sub-time series (2,6,5) in $\textbf{t}_1$ is (1,3,2), since 2 is the smallest value, 6 is the third-smallest value, and 5 is the second-smallest value. Similarly, the sub-time series (5,9,8) in $\textbf{t}_1$ is also (1,3,2). Thus, the number of occurrences of (1,3,2), as an OPP, is 2 in $\textbf{t}_1$.}
			\label{database}
		\end{figure}
		
		
		
		We note that contrast pattern mining {\cite{cpmsurvey,li2021cpma,li2020cont,wu2021topk}} (or distinguishing pattern mining {\cite{wu2019mini}} or discriminative pattern mining {\cite{he2019sign,he2017cond}}) can discover patterns with significant differences in different classes, which means that this approach is used to mine patterns that occur frequently in one class and infrequently in the other. The mined patterns can be used as the features for a sequence classification task \cite{smedt2020tkde}, and they can improve the classification accuracy and the interpretability of the classification model since these patterns have good discrimination. For example, the OPP of the red parts in Fig. \ref{database} is (1,3,2) and pattern (1,3,2) only occurs in $D_+$, not in $D_-$. Therefore, it is easy to classify the time series using pattern (1,3,2), which is a contrast pattern. This example illustrates that features based on contrast patterns will yield a better classification performance than features based on frequent patterns.



		Inspired by OPP mining, this paper proposes COPP mining to discover contrast patterns that can be used as features for time series classification. Generally, users should set two thresholds, the frequent threshold and the infrequent threshold, to mine the contrast patterns. However, it is difficult for users to set the appropriate thresholds without prior knowledge. To solve this issue, we address the top-\textit{k} contrast pattern mining, which discovers the top-\textit{k} patterns with the largest contrasts between positive and negative classes. The main contributions  are as follows. 
		
		1) To extract features for time series classification, we explore COPP mining and propose the COPP-Miner algorithm to discover the top-\textit{k} COPPs.
		
		2) COPP-Miner is composed of three parts: extreme point extraction (EPE) to reduce the length of the  time series, forward mining and reverse mining to mine the top-\textit{k} COPPs. 
		
		3) Both forward mining and reverse mining contain three steps: group pattern fusion strategy to generate candidate patterns,  support rate calculation (SRC) method to efficiently calculate the support of a pattern, and pruning strategies to further prune candidate patterns.
		
		4) Experimental results on real time series databases show that COPP-Miner has a better running performance than other competitive algorithms. More importantly, the mined patterns can be used as features to obtain a better classification performance than other methods.
		
		
		The rest of this paper is organized as follows. Section \uppercase\expandafter{\romannumeral2} introduces the related work. Section \uppercase\expandafter{\romannumeral3} defines the research problems. Section \uppercase\expandafter{\romannumeral4} proposes the COPP-Miner algorithm. 
		Section \uppercase\expandafter{\romannumeral5} carries out a comparative experimental analysis. Section \uppercase\expandafter{\romannumeral6} gives the conclusion of this paper.
		
\section{Related Work}
Sequential pattern mining (SPM) is an important knowledge discovery method {\cite{wu2022thee,li2021apinm}} and has been widely applied in the fields of social security fraud detection {\cite{okolica2020sequ}}, intelligent life detection {\cite{gan2019asur}}, and time series analysis {\cite{safari2022atyp}}. Moreover, if a sequence is added with time information, it forms episode mining \cite {episodepattern, episoderule}.		However, due to the complexity of real life, the limitations of traditional methods have been gradually exposed. Therefore, according to the different needs of users, researchers have put forward a variety of extension methods, such as repetitive SPM {\cite{wu2018nose,truong2019effi,wu2014mini}}, co-occurrence SPM \cite{cooccpattern,mcor}, spatial co-location pattern mining {\cite{wang2022pref,wang2018redu}}, negative SPM {\cite{Dong2019tnnls, wu2022negative,dong2020anef}}, outlying SPM {\cite{wang2020effi}},  high utility SPM {\cite{ishita2022newa,gan2021tkde, gantcyb}}, and high-average utility SPM\cite {weisong2021, wu2021haop,wu2021hanp, Truongtkde, jinkais2020}. However, these SPM methods mainly focus on single-class sequence datasets and cannot be applied to a multi-class sequence database to find patterns with significant differences. Therefore, researchers have put forward some new pattern mining methods, such as contrast SPM {\cite{wu2019min}} (or discriminative SPM {\cite{mathonat2021anyt}}). 

However, most classical data mining algorithms mainly focus on discrete sequences and cannot run directly on time series data. The inherent structural characteristics of a time series, such as its high dimensions and magnanimity, not only reduce the performance of common data mining algorithms but also affect the accuracy and effectiveness of the final pattern mining results. Therefore, it is necessary to preprocess the original time series to convert them into other forms using two approaches: piecewise linear representation and symbolic representation. Traditional piecewise linear representation algorithms mainly perform dimensionality reduction. For example, Keogh et al. {\cite{keogh2002anon}} proposed a bottom-up piecewise linearization method. Symbolic representation is another general method. For example, Lin et al. {\cite{lin2007expe}} proposed symbolic aggregation approximation  to convert digital data into a symbolic form. However, various kinds of noise are inevitably introduced during the process of converting a time series into a symbol series, which increases the difficulty of pattern matching and makes it difficult to grasp the key trend in the time series. 


Recently, the OPP matching method has been proposed {\cite{cho2015afas}}, which directly uses the relative order of numerical values to represent the rank of an element in a time series without mathematical transformation or symbolization, and it intuitively represents the fluctuation trend in a time series. Inspired by OPP matching, Wu et al. {\cite{wu2022oppm}} proposed an OPP-Miner algorithm, paying more attention to the fluctuation trend of the time series during the mining process, to find all frequent sub-time series with the same OPPs. To further improve the efficiency of the algorithm, Wu et al. {\cite{wu2022oprm}} proposed a pattern mining algorithm for order-preserving rules. During the mining process, the matching results of sub-patterns were used to calculate the support of the super-pattern, which greatly improved the efficiency of the algorithm. 

However, it is difficult to form classification features for time series of different classes  using OPP mining \cite {wu2022oppm} and OPR mining \cite{wu2022oprm}. Therefore, we introduce the concept of contrast pattern, which refers to patterns that occur frequently in one class and infrequently in the other. These patterns help to improve not only the classification performance but also the interpretability of the classification model. Traditional contrast pattern mining {\cite{yang2015mini}} requires the user to set two thresholds, the frequent  and infrequent thresholds. However, it is difficult for users to set these thresholds for feature extraction. To tackle these problems, we focus on discovering the top-\textit{k} COPPs for time series classification.

\section{Problem Definition}

A time series with length \textit{n}, expressed as $\textbf{t} = (t_1,\ldots,t_i,\ldots,t_n)(1 \leqslant i \leqslant n)$, is a series of observed numerical values indexed in time order, where $t_i$ is a numerical value. Suppose  a time series binary classification database \textit{D} has \textit{x} positive time series and (\textit{y}-\textit{x}) negative time series, i.e., $D = \{ {\textbf{t}}_1, {\textbf{t}}_2, \ldots , {\textbf{t}}_{x}, {\textbf{t}}_{x+1}, \ldots , {\textbf{t}}_y\},\ D_+ = \{{\textbf{t}}_1, {\textbf{t}}_2, \ldots , {\textbf{t}}_{x}\},$ and $D_- = \{ {\textbf{t}}_{x+1}, \ldots , {\textbf{t}}_y\},$ where  $\textbf{t}_{j} \ (1 \leqslant j \leqslant y)$ is a time series.

\begin{definition}\label{definition1}
	The relative order of a value $v_{i}$ in a group of values \textbf{v} is its rank order, denoted as $R_{\textbf{v}}(v_i)$
\end{definition}

\begin{definition}\label{definition2}
	A pattern represented by a group of relative orders is called an order-preserving pattern (OPP), denoted as $P(\textbf{p}) = (R_{\textbf{p}}(p_{1}),R_{\textbf{p}}(p_{2}),\ldots,R_{\textbf{p}}(p_{m}))$.
\end{definition}

\begin{example}\label{example1}
	Suppose that we have a group of numerical values \textbf{v} = (4,2,6,5). We know that $R_{\textbf{v}}$(4) = 2 since element 4 is the second-smallest element. Similarly, $R_{\textbf{v}}$(2) = 1, $R_{\textbf{v}}$(6) = 4, and $R_{\textbf{v}}$(5) = 3. Hence, the OPP of \textbf{v} is $P(\textbf{v})$ = (2,1,4,3).
\end{example}

\begin{definition}\label{definition3}
	Suppose that we have a time series \textbf{t} = ($t_1,t_2,\ldots,t_n$) and a pattern $\textbf{p} = (p_{1},p_{2},\ldots,p_{m})$. If there is a sub-time series \textbf{i}= $(t_{k},t_{k+1},\ldots,t_{k+m-1})$ (1 $\leqslant$ k $\leqslant$ n-m+1) that satisfies $P(\textbf{i}) = P(\textbf{p})$, then \textbf{i} is an occurrence of pattern \textbf{p} in \textbf{t}, and the last position of \textbf{i}, $<$k+m-1$>$, is used to represent this occurrence. The number of occurrences of \textbf{p} in \textbf{t} is the support, denoted as \textit{sup}(\textbf{p}, \textbf{t}).
\end{definition}

\begin{definition}\label{definition4}
	The density of a pattern \textbf{p} in \textbf{t}, denoted as \textit{den}(\textbf{p}, \textbf{t}), is the ratio of its support \textit{sup}(\textbf{p}, \textbf{t}) to the length of \textbf{t}, i.e., \textit{den}(\textbf{p},\textbf{t}) = \textit{sup}(\textbf{p}, \textbf{t})/$|\textbf{t}|$.  If \textit{den}(\textbf{p}, \textbf{t}) is greater than a given density threshold \textit{minden}, then the count \textit{C}(\textbf{p}, \textbf{t}) is 1; otherwise, \textit{C}(\textbf{p}, \textbf{t}) is 0. The support rate of a pattern \textbf{p} in a time series binary classification database \textit{D} is the sum of the ratio of the counts for all time series to the number of time series, denoted as \textit{r}(\textbf{p}, D), i.e., $r(\textbf{p}, D) = \sum_{i=1}^{|D|} \frac{C(\textbf{p}, \textbf{t}_i)}{|D|}$.		
	
\end{definition}

\begin{example}\label{example2}
	Suppose that we have a time series \textbf{t} = (4,2,6,5,9,8), \textbf{p} = (1,3,2), and \textit{minden} = 0.1. There are two sub-time series \textbf{e} = (2,6,5) and \textbf{f} = (5,9,8) whose relative order is \textbf{p}, i.e., $P(\textbf{e}) = P(\textbf{f}) = P(\textbf{p})$. Thus, \textbf{e} and \textbf{f} are two occurrences of \textbf{p} in \textbf{t}, i.e., \textit{sup}(\textbf{p}, \textbf{t}) = 2. According to Definition \ref{definition4}, \textit{den}(\textbf{p}, \textbf{t}) = \textit{sup}(\textbf{p}, \textbf{t})/$|\textbf{t}|$ = 2$/$6 $>$ 0.1. Hence, the count is 1, i.e., \textit{C}(\textbf{p}, \textbf{t}) = 1.
\end{example}

\begin{definition}\label{definition5}
	Suppose that a time series binary classification database \textit{D} is composed of $\textit{D}_+$ and $\textit{D}_-$. The support rates of \textbf{p} in $\textit{D}_+$ and $\textit{D}_-$ are \textit{r}(\textbf{p}, $\textit{D}_+$) and \textit{r}(\textbf{p}, $\textit{D}_-$), respectively. The contrast rate of \textbf{p} in \textit{D} is denoted as \textit{c}(\textbf{p}, \textit{D}), where \textit{c}(\textbf{p}, \textit{D}) = $|\textit{r}(\textbf{p}, \textit{D}_+)-\textit{r}(\textbf{p}, \textit{D}_-)|$.
\end{definition}

\begin{definition}\label{definition6}
	Our goal is to mine the top-\textit{k} COPPs in the time series binary classification database \textit{D} with a density threshold \textit{minden}.
\end{definition}

\begin{table}[!htb]
	\renewcommand{\arraystretch}{1.1}	
	\footnotesize
	\captionsetup{font=footnotesize} 
	\caption{Time series database \textit{D}}
	\centering
	\label{tab1}
	\tabcolsep 12pt 
	\begin{tabular}{ccc}
		\hline\noalign{\smallskip}
		tid & Time series  &Class\\\hline
		1 & 4, 2, 6, 5, 9, 8 & $\textit{D}_+$\\
		2 & 7, 5, 2, 4, 3, 6 & $\textit{D}_+$\\
		3 & 2, 5, 4, 8, 6, 9 & $\textit{D}_+$\\
		4 & 4, 7, 3, 6, 1 & $\textit{D}_-$\\
		5 & 8, 6, 9, 5, 7 & $\textit{D}_-$\\
		6 & 9, 7, 6, 5, 7, 8 & $\textit{D}_-$\\
		\noalign{\smallskip}\hline
	\end{tabular}
\end{table}

\begin{example}\label{example3}
	Suppose that we have a time series binary classification database \textit{D}, shown in Table \ref{tab1}.
	According to Example \ref{example2}, we know that the support of \textbf{p} = (1,3,2) in $\textbf{t}_1$ is 2. Similarly, we know that the supports of \textbf{p} = (1,3,2) for $\textbf{t}_2,\ \textbf{t}_3,\ \textbf{t}_4,\ \textbf{t}_5$, and $\textbf{t}_6$ are 1, 2, 0, 0, and 0, respectively. The lengths of $\textbf{t}_1$ to $\textbf{t}_6$ are 6, 6, 6, 5, 5, and 6, respectively. Therefore, the densities of pattern \textbf{p} for $\textbf{t}_1$ to $\textbf{t}_6$ are 0.33, 0.17, 0.33, 0, 0, and 0, respectively. Suppose that we set \textit{minden} = 0.1. Then, according to Definition \ref{definition5}, $C(\textbf{p}, \textbf{t}_1) = C(\textbf{p}, \textbf{t}_2) = C(\textbf{p}, \textbf{t}_3)$ = 1 and $C(\textbf{p}, \textbf{t}_4) = C(\textbf{p}, \textbf{t}_5) = C(\textbf{p}, \textbf{t}_6)$ = 0. Since $\textbf{t}_1$, $\textbf{t}_2$, and $\textbf{t}_3$ are in $D_+,\ r(\textbf{p}, D_+)$ = 3/3 = 1. Moreover, $r(\textbf{p}, D_-)$ = 0/3 = 0. Hence, $c(\textbf{p}, D) = r(\textbf{p}, D_+)-r(\textbf{p}, D_-)$ = 1.
	Similarly, the densities of \textbf{q} = (2,3,1) for $\textbf{t}_1$ to $\textbf{t}_6$ are 0, 0, 0, 0.4, 0.2, and 0, respectively. Thus, $r(\textbf{q}, D_+)$ = 0 and $r(\textbf{q}, D_-)$ = 0.67. Hence $c(\textbf{q}, D) = r(\textbf{q}, D_-)-r(\textbf{}q, D_+)$ = 0.67. 	
\end{example}

For clarification, the main symbols are listed in Table \ref{tab2}.
\begin{table}[!htb]
	\renewcommand{\arraystretch}{1.1}	
	\footnotesize
	\captionsetup{font=footnotesize} 
	\caption{Notations}
	\centering
	\label{tab2}
	\tabcolsep 5pt 
	\begin{tabular}{cc}
		\hline\noalign{\smallskip}
		Symbol & Description\\\hline
		\textbf{t} & A time series ($t_1, t_2,\ldots, t_n$) with length \textit{n}\\
		\textit{D} & A database containing \textit{y} time series$\{ {\textbf{t}}_1, {\textbf{t}}_2, \ldots , {\textbf{t}}_x, {\textbf{t}}_{x+1}, \ldots , {\textbf{t}}_y\}$\\
		$\textit{D}_+$$/$$\textit{D}_-$ & A positive$/$negative class time series\\
		\textbf{p} & A pattern ($p_{1}, p_{2}, \ldots, p_{m}$) with length \textit{m}\\
		$R_{\textbf{v}}(v_{i})$ & The relative order of value $v_{i}$ in a group of values \textbf{v}\\
		$P(\textbf{p})$ & An order-preserving pattern\\
		\textit{sup}(\textbf{p}, \textbf{t})& Support of pattern \textbf{p} in \textbf{t}\\
		\textit{minden}& A density threshold\\
		\textit{den}(\textbf{p}, \textbf{t})& The density of pattern \textbf{p} in \textbf{t}\\
		\textit{r}(\textbf{p}, \textit{D})& The support rate of pattern \textbf{p} in \textit{D}\\
		\textit{c}(\textbf{p}, \textit{D})& The contrast rate of pattern \textbf{p} in \textit{D}\\
		\noalign{\smallskip}\hline
	\end{tabular}
\end{table}

\section{Algorithm design}

COPP-Miner has three parts: extreme point extraction (EPE), forward mining, and reverse mining. Section \ref{sub4.1} proposes the EPE algorithm which aims to reduce the length of the time series and compress the size of dataset, so as to improve the mining efficiency of  COPP-Miner. The forward mining discovers the top contrast patterns which are most frequent in the positive class and are infrequent in the negative class, and the reverse mining continues to discover the top contrast patterns which are most frequent in the negative class and are infrequent in the positive class. Section \ref{sub4.2} presents the forward mining which contains three steps: group pattern fusion strategy to generate candidate patterns,  support rate calculation (SRC) method to efficiently calculate the support of a pattern, and pruning strategies to further prune candidate patterns. Section \ref{sub4.3} shows the reverse mining that is the same as forward mining, except for using one pruning strategy.  Finally, we illustrate COPP-Miner in Section \ref{sub4.4}. The framework of COPP-Miner is shown in Fig. \ref{framework}.


\begin{figure}[!htb]
	\centering
	\includegraphics[width=\linewidth]{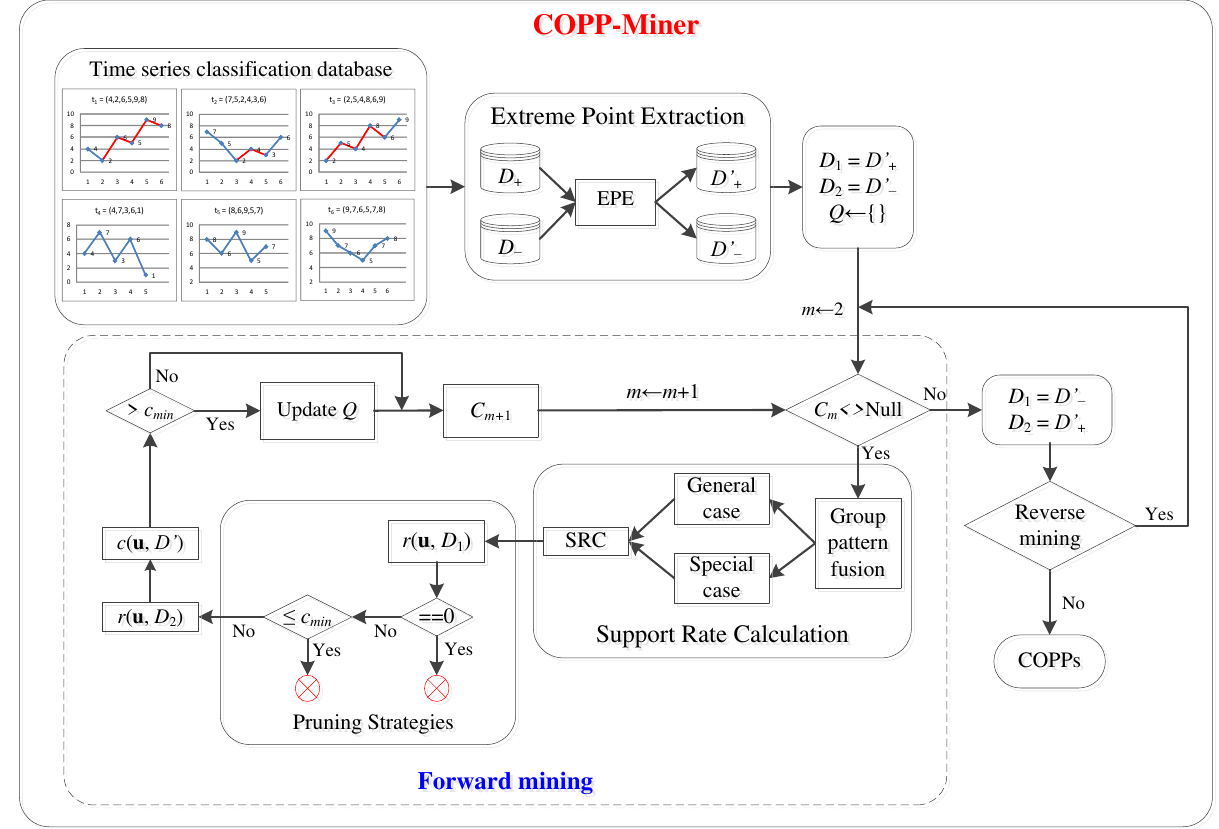}
	\captionsetup{font=footnotesize} 
	\caption{Framework of COPP-Miner which contains three parts: extreme point extraction (EPE) to extract the local extreme points, forward mining, and reverse mining to discover COPPs.  EPE aims to extract local extreme points of time series. Forward mining discovers the patterns which are frequent in $D'_+$ and infrequent in $D'_-$, and contains three steps: group pattern fusion strategy to generate candidate patterns, support calculation method to efficiently calculate the pattern support, and two pruning strategies to further prune candidate patterns. Reverse mining further mines contrast patterns that are frequent in $D'_-$ and infrequent in $D'_+$. Reverse mining is the same as forward mining, except for only using one pruning strategy. }
	\label{framework}		
\end{figure}

\subsection{Extreme Point Extraction }	\label {sub4.1}

The local extreme point method can extract the effective features of a time series, compress the data, and improve the accuracy and efficiency of the algorithm. Hence, we employ local extreme points to determine the natural segmentation points of time series trends.  The extreme points include local minimal points and local maximal points.

\begin{definition}\label{definition7}
	Suppose that we have a time series \textbf{t} with length \textit{n}. If $t_{i}$ satisfies one of the following four conditions, then $t_{i}$ is stored as a local extreme point {\cite{abdelmadjid2021afas}}:
	\begin{itemize}
		\item[•] If \textit{i} = 1, then $t_{i}$ is the first element;
		\item[•] If \textit{i} = \textit{n}, then $t_{i}$ is the last element;
		\item[•] If  $t_i < t_{i-1}$ and $t_i \leqslant t_{i+1}$ or $t_i \leqslant t_{i-1}$ and $t_i < t_{i+1}$, then $t_{i}$ is a local minimal point;
		\item[•] If  $t_{i-1} < t_i$ and $t_{i+1} \leqslant t_i$ or $t_{i-1} \leqslant t_i$ and $t_{i+1} < t_i$, then $t_{i}$ is a local maximal point.
	\end{itemize}
\end{definition}

Based on Definition \ref{definition7}, we use the extreme point extraction (EPE) algorithm to process the original time series and extract the local extreme points to form a new time series. The pseudo-code of the EPE algorithm is omitted.	


\begin{example}\label{example4}
	In this example, we use $\textbf{t}_{6}$ = (9,7,6,5,7,8) in Table \ref{tab1} to illustrate the principle of the EPE algorithm. According to Definition \ref{definition7}, 9 is stored, since it is the first element. 7 is not stored, since it is not a local extreme point. Similarly, 6 is not stored. However, 5 is a local extreme point and is stored, since 5 $<$ 6 and 5 $\leqslant$ 7. Similarly, 7 is not stored. Finally, 8 is stored, since it is the last element. Hence, the new time series is $\textbf{t'}_6$ = (9,5,8).
\end{example}

In Example \ref{example4}, the length of the original time series $\textbf{t}_6$ is 6. After using EPE, the length of the new time series is 3. Hence, EPE can effectively compress a time series.	

\subsection{Forward mining } \label {sub4.2}

According to Definition \ref{definition6}, we know that $c(\textbf{p}, D) = |r(\textbf{p}, D_+)-r(\textbf{p}, D_-)|$, which makes it difficult to mine the top-\textit{k} COPPs. Thus, we mine the top-\textit{k} COPPs whose $c(\textbf{p}, D)$ is $r(\textbf{p}, D_+)-r(\textbf{p}, D_-)$, named forward mining, and we mine the top-\textit{k} COPPs whose $c(\textbf{p}, D)$ is $r(\textbf{p}, D_-)-r(\textbf{p}, D_+)$, named reverse mining. The forward mining has three steps: group pattern fusion strategy to generate candidate patterns, support calculation method to efficiently calculate the pattern support, and pruning strategies to further prune candidate patterns.	

\subsubsection{Candidate pattern generation }
\label {sub4.2.1}

To effectively reduce the number of candidate patterns, a pattern fusion strategy was proposed in \cite{wu2022oppm} which outperforms the enumeration strategy. However, we note that after extracting extreme points, the pattern fusion strategy still contains many unnecessary calculations. To improve the efficiency, we further propose a group pattern fusion strategy. Now, we first introduce the relevant definitions of the pattern fusion strategy. 

\begin{definition}\label{definition8}
	Suppose that we have a list of values $\textbf{p} = (p_{1},p_{2},\ldots,p_{m})$. $\textbf{r} = P((p_{1}, p_{2}, …, p_{m-1}))$ is called the prefix OPP of \textbf{p}, denoted as \textbf{r} = \textit{prefixop}(\textbf{p}), and $\textbf{f} = P((p_2, p_3, …, p_m))$ is called the suffix OPP of \textbf{p}, denoted as \textbf{f} = \textit{suffixop}(\textbf{p}).
\end{definition}	

\begin{definition}\label{definition9}
	Suppose that we have two patterns \textbf{p} = $(p_{1},p_{2},\ldots,p_{m})$ and \textbf{q} = $(q_{1},q_{2},\ldots,q_{m})$. If \textit{suffixop}(\textbf{p}) = \textit{prefixop}(\textbf{q}), then \textbf{p} and \textbf{q} can generate a super-pattern with length \textit{m} + 1. By comparing the values of $p_1$ and $q_m$, this fusion can be divided into a general case and a special case.
\end{definition}

\textbf{General case}: If $p_{1} \neq q_{m}$, then \textbf{p} and \textbf{q} can generate  a pattern with a length of \textit{m} + 1, \textbf{u} = $(u_{1},u_{2},\ldots,u_{m+1})$, denoted as \textbf{u} = \textbf{p} $\oplus$ \textbf{q}. $u_i$ can be calculated as follows:

\textcircled{1} If $p_{1} < q_{m}$, then $u_{1} = p_{1}$ and $u_{m+1} = q_{m}+1$. Moreover, if $p_{i} < u_{m+1}$, then $u_{i} = p_{i}$, otherwise, $u_{i} = p_{i}+1 (2\leqslant i \leqslant m)$.

\textcircled{2} If $p_{1} > q_{m}$, then $u_{1} = p_{1}+1$ and $u_{m+1} = q_{m}$. Moreover, if $q_{j} < u_{1}$, then $u_{j+1} = q_{j}$, otherwise, $u_{j+1} = q_{j}+1 (1\leqslant j \leqslant m-1)$.

\textbf{Special case}: If $p_{1} = q_{m}$, then \textbf{p} and \textbf{q} can generate two patterns with length of \textit{m} + 1, \textbf{u} = $(u_{1},u_{2},\ldots,u_{m+1})$ and \textbf{v} = $(v_{1},v_{2},\ldots,v_{m+1})$ , denoted as \textbf{u}, \textbf{v} = \textbf{p} $\oplus$ \textbf{q} where $u_{1} = v_{m+1} = p_{1}$ and $u_{m+1} = v_{1} = p_{1}+1$, and if $p_{i} < p_{1}$, then $u_{i} = v_{i} = p_{i}$, otherwise, $u_{i} = v_{i} = p_{i}+1 (2\leqslant i \leqslant m)$.

\begin{example}\label{example5}
	Suppose that we have two patterns \textbf{p} = (2,1,3) and \textbf{q} = (1,3,2), and the pattern fusion strategy is adopted to generate candidate pattern set $C_4$. 
	
	If \textbf{p} is used as a prefix pattern, \textbf{q} as a suffix pattern, since \textit{suffixop}(\textbf{p}) = \textit{prefixop}(\textbf{q}) = (1,2), then pattern \textbf{p} can fuse with pattern \textbf{q} to generate candidate patterns shown as follows:
	
	\begin{enumerate}
		\item[(i)] \textbf{p} $\oplus$ \textbf{q}. It belongs to the special case, since \textit{suffixop}(\textbf{p}) = \textit{prefixop}(\textbf{q}) = (1,2), and $p_1$ = 2 is equal to $q_3$ = 2. Hence, we generate two candidate patterns: $u_1$ = $v_4$ = 2, $u_4$ = $v_1$ = 2 + 1 = 3, $p_2 < p_1$, $u_2$ = $v_2$ = 1, $p_3 > p_1$, and $u_3$ = $v_3$ = 3 + 1 = 4. Thus, \textbf{u} = (2,1,4,3) and \textbf{v} = (3,1,4,2).			
	\end{enumerate}	
	
	If \textbf{q} is used as a prefix pattern, \textbf{p} as a suffix pattern, since \textit{suffixop}(\textbf{q}) = \textit{prefixop}(\textbf{p}) = (2,1), then pattern \textbf{q} can fuse with pattern \textbf{p} to generate candidate patterns shown as follows:
	
	\begin{enumerate}
		\item[(ii)] \textbf{q} $\oplus$ \textbf{p}. It belongs to the general case, since \textit{suffixop}(\textbf{q}) = \textit{prefixop}(\textbf{p}) and $q_1$ = 1 is less than $p_3$ = 3. Hence, $u_1$ = 1, $u_4$ = 3 + 1 = 4, $q_2 < u_4$, $u_2$ = $q_2$ = 3, $q_3 < u_4$, and $u_3$ = $q_3$ = 2. Thus, \textbf{u} = (1,3,2,4).
	\end{enumerate}	
	
\end{example}


However, the pattern fusion strategy still produces many unnecessary calculations. To reduce redundant calculations, this paper proposes a group pattern fusion strategy. The principle is as follows.

{All frequent patterns and candidate patterns are  divided into two parts: Group-1 and Group-2. Patterns  (1,2) and (2,1) are stored in Group-1 and Group-2, respectively. There are three cases for \textbf{p} $\oplus$ \textbf{q}. }

{Case 1:  Patterns \textbf{p} and \textbf{q} are in Group-1 and Group-2, respectively. In this case, the generated candidate patterns are stored in Group-1.}

{Case 2:  Patterns \textbf{p} and \textbf{q} are in Group-2 and Group-1, respectively. In this case, the generated candidate patterns are stored in Group-2.}

{Case 3:  Patterns \textbf{p} and \textbf{q} are in Group-1 or Group-2. In this case, the generated candidate patterns can be pruned.}

\begin{theorem}\label {theorem group}
	{Suppose that patterns  \textbf{p} and  \textbf{q} are in the same group. In this case, pattern \textbf{p} does not need to fuse with pattern  \textbf{q}, which means that the group pattern fusion strategy is complete. }
\end{theorem}
\begin{proof}
	{ 
		It is easy to know that the relative order of the first two values of each pattern in Group-1 is (1,2), and that in Group-2 is (2,1). If patterns  \textbf{p} and  \textbf{q} are in Group-1, then the relative orders of the first two values of \textbf{p} and  \textbf{q} are both (1,2). Therefore, the relative order of the first three values of \textbf{u} is (1,2,3), where  \textbf{u} = \textbf{p} $\oplus$ \textbf{q}.
		Similarly, if  patterns  \textbf{p} and  \textbf{q} are in Group-2, then the relative orders of the first three values of \textbf{u} is (3,2,1). However, according to EPE, the shrunk time series does not have the sub-time series whose relative order is (1,2,3) or (3,2,1). Thus, pattern \textbf{u} can be pruned, which means that if patterns  \textbf{p} and  \textbf{q} are in the same group, then pattern \textbf{p} does not need to fuse with pattern  \textbf{q}. Hence, the group pattern fusion strategy is complete. }
	\end {proof}
	
	{The advantage of the group pattern fusion strategy is as follows. Suppose that there are $l$ candidate patterns with length $m$, and each group has about $l/2$ candidate patterns. If we adopt the pattern fusion strategy to generate candidate patterns, the method will check $l \times l$ times, since each pattern can fuse with all patterns. However, if we adopt the group pattern fusion strategy to generate candidate patterns, the method will fuse $l \times l/2$ times, since each pattern in Group-1 can fuse with all patterns in Group-2, and each pattern in Group-2 can fuse with all patterns in Group-1. Thus, each process fuses $l \times l/4$ times. Hence, the group pattern fusion strategy can reduce the times of pattern fusion. }

	\subsubsection{Support rate calculation}
	\label {sub4.2.2}
	
	To avoid redundant support calculations, we propose a support rate calculation (SRC) algorithm, which calculates the support rate of super-patterns based on the matching results of sub-patterns. The principle of this algorithm is illustrated as follows.
	
	Suppose that we know all the occurrences of each pattern with a length of \textit{m} (\textit{m} $\geqslant$ 2), and all occurrences of pattern \textbf{p} in a time series \textbf{t} and in a time series dataset \textit{D} are stored in sets $L_{\textbf{p}, \textbf{t}}$ and $L_{\textbf{p}, \textit{D}}$, respectively. Thus, the size of $L_{\textbf{p}, \textbf{t}}$ is the support \textit{sup}(\textbf{p}, \textbf{t}), i.e., \textit{sup}(\textbf{p}, \textbf{t}) = $|L_{\textbf{p}, \textbf{t}}|$. We use two auxiliary sets $\mathcal{P}_{\textbf{p}, \textbf{t}}$ and $\mathcal{S}_{\textbf{p}, \textbf{t}}$ to store the prefix set and the suffix set of pattern \textbf{p}, and we use them to calculate all occurrences of their super-patterns with a length of \textit{m} + 1. At the beginning of the procedure, $\mathcal{P}_{\textbf{p}, \textbf{t}}$ and $\mathcal{S}_{\textbf{p}, \textbf{t}}$ are the same as $L_{\textbf{p}, \textbf{t}}$. If we generate the candidate super-patterns using pattern fusion (\textbf{p} $\oplus$ \textbf{q}), then we dynamically delete the elements in $\mathcal{P}_{\textbf{p}, \textbf{t}}$ and $\mathcal{S}_{\textbf{p}, \textbf{t}}$. The details of this process follow.
	
	\textbf{Rule 1. General case.} According to Definition \ref{definition9}, if $p_{1}\neq q_{m}$, then \textbf{p} and \textbf{q} generate a pattern \textbf{u} = \textbf{p} $\oplus$ \textbf{q}. Moreover, if $l_{q_j} = l_{p_i}$ + 1, ($l_{p_i}$ $\in$ $\mathcal{P}_{\textbf{p}, \textbf{t}}$ and $l_{q_j}$ $\in$ $\mathcal{S}_{\textbf{p}, \textbf{t}}$), then $l_{q_j}$ is an occurrence of \textbf{u}, $l_{q_j}$ is added into $L_{\textbf{u}, \textbf{t}}$, and $l_{p_i}$ and $l_{q_j}$ are removed from $\mathcal{P}_{\textbf{p}, \textbf{t}}$ and $\mathcal{S}_{\textbf{p}, \textbf{t}}$, respectively.
	
	\textbf{Rule 2. Special case.} According to Definition \ref{definition9}, if $p_1 = q_m$, then \textbf{p} and \textbf{q} generate two patterns \textbf{u}, \textbf{v} = \textbf{p} $\oplus$ \textbf{q}. Moreover, if $l_{q_j} = l_{p_i}$ + 1, ($l_{p_i}$ $\in$ $\mathcal{P}_{\textbf{p}, \textbf{t}}$ and $l_{q_j}$ $\in$ $\mathcal{S}_{\textbf{p}, \textbf{t}}$), then $l_{q_j}$ may be an occurrence of \textbf{u} or \textbf{v}. Now, we compare the values of $t_{first}$ and $t_{last}$ in time series \textbf{t}, where \textit{first} = $l_{q_j}$ - \textit{m} and \textit{last} = $l_{q_j}$. If $t_{first} < t_{last}$, then $l_{q_j}$ is an occurrence of \textbf{u}, i.e., $l_{q_j}$ is added into $L_{\textbf{u}, \textbf{t}}$, and $l_{p_i}$ and $l_{q_j}$ are removed from $\mathcal{P}_{\textbf{p}, \textbf{t}}$ and $\mathcal{S}_{\textbf{p}, \textbf{t}}$, respectively. If $t_{first} > t_{last}$, then $l_{q_j}$ is an occurrence of \textbf{v}, i.e., $l_{q_j}$ is added into $L_{\textbf{v}, \textbf{t}}$, and $l_{p_i}$ and $l_{q_j}$ are removed from $\mathcal{P}_{\textbf{p}, \textbf{t}}$ and $\mathcal{S}_{\textbf{p}, \textbf{t}}$, respectively. If $t_{first} = t_{last}$, then $l_{q_j}$ is not an occurrence of \textbf{u} or \textbf{v}.
	
	Example \ref{example6} illustrates the principle of Rules 1 and 2.
	
	\begin{example}\label{example6}
		Take $\textbf{t}_1$ = (4,2,6,5,9,8) in Table \ref{tab1} as an example. Suppose that we have two patterns \textbf{p} = (2,1,3) and \textbf{q} = (1,3,2). It is easy to see that $L_{\textbf{p}, \textbf{t}_1}$ = \{3,5\} and $L_{\textbf{q}, \textbf{t}_1}$ = \{4,6\}. Now, we calculate the supports of the super-patterns \textbf{p} $\oplus$ \textbf{q}.

		According to Definition \ref{definition9}, we know that \textbf{p} $\oplus$ \textbf{q} belongs to the special case. Thus, we generate two super-patterns, \textbf{u} = (2, 1, 4, 3) and \textbf{v} = (3, 1, 4, 2). At the beginning of the process, $\mathcal{P}_{\textbf{p}, \textbf{t}_1}$ = $L_{\textbf{p}, \textbf{t}_1}$ and $\mathcal{S}_{\textbf{q}, \textbf{t}_1}$ = $L_{\textbf{q}, \textbf{t}_1}$, since \textbf{p} and \textbf{q} are used as prefix and suffix patterns, respectively. We know that 3 $\in$  $\mathcal{P}_{\textbf{p}, \textbf{t}_1}$ and 3 + 1 = 4 $\in$ $\mathcal{S}_{\textbf{q}, \textbf{t}_1}$. Now, we determine whether 4 is an occurrence of \textbf{u} or \textbf{v} using Rule 2. We know that the length of pattern \textbf{p} is 3. Thus, \textit{first} = 4 - 3 = 1, and \textit{last} = 4. Since $t_1$ = 4 $< t_4$ = 5, 4 is an occurrence of \textbf{u}, i.e., 4 is added into $L_{\textbf{u}, t_1}$, and 3 and 4 are removed from $\mathcal{P}_{\textbf{p}, \textbf{t}_1}$ and $\mathcal{S}_{\textbf{q}, \textbf{t}_1}$, respectively. Similarly, we know that $L_{\textbf{u}, t_1}$ = $\{$4, 6$\}$ and $L_{\textbf{v}, t_1}$ = $\phi$, i.e., $sup(\textbf{u}, \textbf{t}_1)$ = 2 and $sup(\textbf{v}, \textbf{t}_1)$ = 0. Moreover, $\mathcal{P}_{\textbf{p}, \textbf{t}_1} = \phi$ and $\mathcal{S}_{\textbf{q}, \textbf{t}_1} = \phi$.	
	\end{example}
	
	The advantage of dynamically deleting the elements in $\mathcal{P}_{\textbf{p}, \textbf{t}_1}$ and $\mathcal{S}_{\textbf{q}, \textbf{t}_1}$ follows. The sizes of $\mathcal{P}_{\textbf{p}, \textbf{t}_1}$ and $\mathcal{S}_{\textbf{q}, \textbf{t}_1}$ decrease gradually, since the elements in $\mathcal{P}_{\textbf{p}, \textbf{t}_1}$ and $\mathcal{S}_{\textbf{q}, \textbf{t}_1}$ are being dynamically deleted. If we continue to use $\mathcal{P}_{\textbf{p}, \textbf{t}_1}$ or $\mathcal{S}_{\textbf{q}, \textbf{t}_1}$ to calculate the supports of their super-patterns, the efficiency of the algorithm will be improved. Example \ref{example7} illustrates the advantage of this method.
	
	\begin{example}\label{example7}
		Suppose that we have two patterns \textbf{p} = (2,1,3) and \textbf{q} = (2,3,1). In this example, we use the results of Example \ref{example6} to calculate the support of the super-pattern \textbf{p} $\oplus$ \textbf{q}.

		According to Definition \ref{definition9}, we know that \textbf{p} $\oplus$ \textbf{q} belongs to the general case and generates one super-pattern \textbf{w} = (3,2,4,1). According to Example \ref{example6}, we know that $\mathcal{P}_{\textbf{p}, \textbf{t}_1} = \phi$, which means that there is no occurrence $l_{p_i}$ in $\mathcal{P}_{\textbf{p}, \textbf{t}_1}$. Therefore, there is no occurrence for \textbf{p} $\oplus$ \textbf{q}. Hence, $L_{\textbf{w}, \textbf{t}_1} = \phi$, i.e., \textit{sup}(\textbf{w}, $\textbf{t}_1$) = 0 can be obtained directly. Thus, this example shows that this method can effectively improve the efficiency of the support rate calculation.
		
	\end{example}
	
	The support of a pattern \textbf{p} in a time series \textbf{t}, i.e., \textit{sup}(\textbf{p}, \textbf{t}), can be calculated using the above method. We can further calculate \textit{den}(\textbf{p}, \textbf{t}), \textit{C}(\textbf{p}, \textbf{t}), and \textit{r}(\textbf{p}, \textit{D}) according to Definition \ref{definition5}. The SRC algorithm is shown in Algorithm \ref{Algorithm SRC}.

	\begin{algorithm}[htb]\label{Algorithm SRC}
		\caption{SRC: Support rate calculation in time series dataset \textit{D}}	
		\hspace*{0.02in} \leftline{{\bf Input:}
			Pattern \textbf{p} and its prefix set $\mathcal{P}_{\textbf{p}, \textit{D}}$, pattern \textbf{q} and its suffix set $\mathcal{S}_{\textbf{q}, \textit{D}}$, and a density threshold \textit{minden} }
		\hspace*{0.02in} \leftline{{\bf Output:}
			\textit{r}(\textbf{u}, \textit{D}) and its $L_{\textbf{u}, \textit{D}}$, \textit{r}(\textbf{v}, \textit{D}) and its $L_{\textbf{v}, \textit{D}}$, $\mathcal{P}_{\textbf{p}, \textit{D}}$ and $\mathcal{S}_{\textbf{q}, \textit{D}}$ }
		\begin{algorithmic}[1]
			\State \textbf{f} $\leftarrow$ \textit{suffixop}(\textbf{p});  
			\State \textbf{r} $\leftarrow$ \textit{prefixop}(\textbf{q});
			\If {\textbf{f} == \textbf{r} }	
			\If {\textbf{p}[0] == \textbf{q}[\textit{m}-1] }
			\State \textbf{u} $\cup$ \textbf{v}  $\leftarrow$ \textbf{p} $\oplus$ \textbf{q}; 
			\For {each \textbf{t} in \textit{D}}
			\State Obtain the matching set $L_{\textbf{u}, \textbf{t}}$ and $L_{\textbf{v}, \textbf{t}}$, and update $\mathcal{P}_{\textbf{p}, \textbf{t}}$ and $\mathcal{S}_{\textbf{q}, \textbf{t}}$ according to Rule 1;
			\State Update \textit{r}(\textbf{u}, \textit{D}) and \textit{r}(\textbf{v}, \textit{D}) according to Definition \ref{definition5};
			\EndFor
			\Else
			\State \textbf{u} $\leftarrow$ \textbf{p}$\oplus$\textbf{q}; 
			\For {each \textbf{t} in \textit{D}}
			\State Obtain the matching set $L_{\textbf{u}, \textbf{t}}$, and update $\mathcal{P}_{\textbf{p}, \textbf{t}}$ and $\mathcal{S}_{\textbf{q}, \textbf{t}}$ according to Rule 2;
			\State Update \textit{r}(\textbf{u}, \textit{D}) according to Definition \ref{definition5};
			\EndFor
			\EndIf
			\EndIf
			\State \Return \textit{r}(\textbf{u}, \textit{D}) and its $L_{\textbf{u}, \textit{D}}$, \textit{r}(\textbf{v}, \textit{D}) and its $L_{\textbf{v}, \textit{D}}$, $\mathcal{P}_{\textbf{p}, \textit{D}}$ and $\mathcal{S}_{\textbf{q}, \textit{D}}$;
		\end{algorithmic}
	\end{algorithm}

	\subsubsection{Pruning strategies}
	\label {sub4.2.3}
	
	We know that if the mining problem satisfies anti-monotonicity, then we can adopt the Apriori strategy to prune the candidate patterns. Otherwise, we have to design other strategies to prune the candidate patterns. It is easy to know that COPP mining does not satisfy anti-monotonicity. Thus, we propose two strategies to effectively prune the candidate patterns.

	
	

	Since OPP mining in single-class time series satisfies anti-monotonicity \cite {wu2022oppm, wu2022oprm}, it follows that $r(\textbf{p}, D_+)$ satisfies anti-monotonicity. Based on these characteristics, we design two efficient pruning strategies to prune candidate patterns.
	
	\textbf{Pruning Strategy 1.} If $r(\textbf{p}, D_+) = 0$, then pattern \textbf{p} and its super-pattern \textbf{u} will be pruned.
	
	\textbf{Pruning Strategy 2.} Suppose that there are \textit{k} COPPs and the minimal contrast rate  $c_{min} > 0$. If $r(\textbf{p}, D_+) \leqslant c_{min}$, then pattern \textbf{p} and its super-pattern \textbf{u} will be pruned.
	
	{According to Pruning Strategy 2, we know that the greater the $c_{min}$, the more candidate patterns can be pruned. The greater the support rate of a pattern, the greater the possibility of its contrast ratio. Therefore, based on the above support calculation method, we propose a heuristic strategy.}
	
	{\textbf{Support Maximum-first strategy.} According to the group pattern fusion strategy, we know all frequent patterns are divided into two parts. These patterns in each group are in descending order according to the support rate in $D_+$.}
	
	
	To prove the correctness of this support maximum-first strategy, we initially prove the correctness of these pruning strategies in Theorems \ref{theorem1} and \ref{theorem2}.
	
	\begin{theorem}\label{theorem1}
		If $r(\textbf{p}, D_+)$ = 0, then \textbf{p} and its super-pattern \textbf{u} will be pruned.
	\end{theorem}
	
	\begin{proof}
		According to Definition \ref{definition5}, 0 $\leqslant r(\textbf{p}, D_+) \leqslant$ 1 and 0 $\leqslant r(\textbf{p}, D_-) \leqslant$ 1. Since $r(\textbf{p}, D_+)$ = 0, we know that $c(\textbf{p}, D) = r(\textbf{p}, D_+) - r(\textbf{p}, D_-)$ = 0 $ - r(\textbf{p}, D_-) \leqslant$ 0. If \textbf{p} is a top-\textit{k} COPP, then $c(\textbf{p}, D) >$ 0. Therefore, pattern \textbf{p} cannot be a COPP. Assume that pattern \textbf{u} is a super-pattern of \textbf{p}. We know that $r(\textbf{u}, D_+)$ = 0 since $r(\textbf{p}, D_+)$ satisfies anti-monotonicity and $r(\textbf{p}, D_+)$ = 0. Similarly, the super-pattern \textbf{u} cannot be a top-\textit{k} COPP. Hence, \textbf{p} and its super-pattern \textbf{u} can be pruned if $r(\textbf{p}, D_+)$ = 0. 
	\end{proof}
	
	\begin{theorem}\label{theorem2}
		If $r(\textbf{p}, D_+) \leqslant c_{min}$, then pattern \textbf{p} and its super-pattern \textbf{u} will be pruned.
	\end{theorem}
	
	\begin{proof}
		We know that $c(\textbf{p}, D) = r(\textbf{p}, D_+) - r(\textbf{p}, D_-)$ is no greater than $c_{min}$, since 0 $\leqslant r(\textbf{p}, D_-) \leqslant$ 1 and $r(\textbf{p}, D_+) \leqslant c_{min}$. Hence, pattern \textbf{p} is not a top-\textit{k} COPP. Moreover, if \textbf{u} is a super-pattern of pattern \textbf{p}, then $r(\textbf{u}, D_+) \leqslant r(\textbf{p}, D_+)$ since $r(\textbf{p}, D_+)$ satisfies anti-monotonicity. Therefore, $c(\textbf{u}, D) \leqslant c_{min}$. Hence, super-pattern \textbf{u} is not a top-\textit{k} COPP, either.
	\end{proof}

	
	\begin{theorem}\label{theoremMaxFirst}
		The support maximum-first strategy is correct.
	\end{theorem}
	
	\begin{proof}
		It is easy to know that the support maximum-first strategy is complete, since this strategy does not prune any patterns, but only adjusts the order of patterns. We have therefore proved the correctness of the support maximum-first strategy.
	\end{proof}

	\begin{theorem}\label{theoremmax}
		The support maximum-first strategy is complete.
	\end{theorem}
	\begin{proof}
		
		The proof is by contradiction. Assume that the support maximum-first strategy is incomplete which means that there is a pattern \textbf{t} whose contrast rate is greater than $c_{min}$, i.e., $c(\textbf{t}, D)$ $\textgreater$ $c_{min}$, but using the support maximum-first strategy, COPP-Miner cannot mine pattern \textbf{t}. Assume that \textbf{t} = \textbf{p} $\oplus$ \textbf{q}. We know that $r(\textbf{t}, D_+)$ $\textgreater$ $c_{min}$, since $r(\textbf{t}, D_-) \geqslant 0$, $c(\textbf{t}, D) = r(\textbf{t}, D_+) - r(\textbf{t}, D_-)$, and $c(\textbf{t}, D)$ $\textgreater$ $c_{min}$. We know that OPP mining satisfies anti-monotonicity {\cite{wu2022oppm,wu2022oprm}}. Thus, $r(\textbf{p}, D_+) \geqslant r(\textbf{t}, D_+)$ $\textgreater$ $c_{min}$ and $r(\textbf{q}, D_+) \geqslant r(\textbf{t}, D_+)$ $\textgreater$ $c_{min}$. Therefore, both patterns \textbf{p} and \textbf{q} cannot be pruned by Pruning Strategy 2. Using the group pattern fusion strategy, COPP-Miner can generate candidate pattern \textbf{t}, since \textbf{t} = \textbf{p} $\oplus$ \textbf{q}. Hence, COPP-Miner can discover pattern \textbf{t}, which contradicts the assumption that COPP-Miner cannot mine pattern \textbf{t}. Theorem \ref{theoremmax} is proved.

	\end{proof}
	
	\subsection{Reverse mining}
	\label {sub4.3}
	
	In this section, we exchange $D_+$ and $D_-$, and then we only use the condition of $r(\textbf{p}, D_-) \leqslant c_{min}$ and do not use the condition of $r(\textbf{p}, D_-)$ = 0 to prune candidate patterns.
	
	\begin{theorem}\label{theoremgreater}
		After forward mining, $c_{min}$ is greater than zero.
	\end{theorem}
	
	\begin{proof}
		Proof by contradiction. Assuming that after forward mining, $c_{min}$ is equal to zero. We know that $c_{min}$ is equal to zero which means that we do not find $k$ contrast patterns. If we do not find $k$ contrast patterns, the forward mining cannot terminate, which contradicts the fact that $k$ contrast patterns are found in forward mining.
		\end {proof}
		
		According to Theorem \ref {theoremgreater}, after forward mining, $c_{min}>$0. Therefore, it is useless to prune candidate patterns by employing the condition of $r(\textbf{p}, D_-)$ = 0. Hence, in reverse mining, we only use the condition of  $r(\textbf{p}, D_-) \leqslant c_{min}$ to prune candidate patterns.
		
		
		
		\textbf{Pruning Strategy 3.} If $r(\textbf{p}, D_-) \leqslant c_{min}$, then pattern \textbf{p} and its super-pattern \textbf{u} will be pruned.

		We know that Pruning Strategy 3  corresponds to Pruning Strategy 2 in forward mining, and Theorem \ref {theorem2} proves the correctness of Pruning Strategy 2. Therefore, it is easy to know  the correctness of Pruning Strategy 3.
		
		\subsection{COPP-Miner}
		\label {sub4.4}
		
		The COPP-Miner algorithm contains the following steps.
		
		Step 1. Use the EPE algorithm to extract the local extreme points in database \textit{D} to obtain a new dataset $D’$. Mine the top-\textit{k} COPPs whose $c(\textbf{p}, D’)$ is $r(\textbf{p}, D’_+) - r(\textbf{p}, D’_-)$ in forward mining.
		
		Step 2. Calculate all occurrences of patterns (1,2) and (2,1) in $D’$, then add these patterns into the candidate sets $C_{1,2}$ and $C_{2,2}$, respectively.
		
		Our mining method starts with a 2-length pattern, since a pattern with length one is only a point and has no fluctuation trend. There are two patterns with length two: (1,2) and (2,1). Scan time series \textbf{t} once. If the relative order of a sub-time series ($t_i$, $t_{i+1}$) is (1,2), then $<i$+1$>$, as an occurrence of pattern (1,2), is stored in set $L_{(1,2), \textbf{t}}$. Otherwise, $<i$+1$>$, as an occurrence of pattern (2,1), is stored in set $L_{(2,1), \textbf{t}}$. 
		
		
		Step 3. Select two patterns \textbf{p} and \textbf{q} from two different groups. According to the group pattern fusion strategy, two super-patterns \textbf{u} and \textbf{v} with lengths of \textit{m} + 1 will be generated. Then, we use the SRC algorithm to calculate $r(\textbf{u}, D’)$ and its $L_{\textbf{u}, D’}$, and $r(\textbf{v}, D’)$ and its $L_{\textbf{v}, D’}$. If pattern \textbf{u} or \textbf{v} cannot be pruned by Pruning strategies 1 and 2, then pattern \textbf{u} or  \textbf{v} will be added into the candidate set $C_{m+1}$. If pattern \textbf{u} or \textbf{v} is a top-\textit{k} COPP, then the top-\textit{k} COPP set \textit{Q} will be updated.
		
		Step 4. Iterate Step 3 until $C_{m+1}$ is empty. 
		
		Step 5. Update the top-\textit{k} COPPs whose $c(\textbf{p}, D’)$ is $r(\textbf{p}, D’_-)-r(\textbf{p}, D’_+)$ in reverse mining. Use the existing top-\textit{k} COPP set \textit{Q}, and iterate Steps 2 to 4 and only use Pruning strategy 3 to prune candidate patterns, until no new patterns are added into the top-\textit{k} set.
		
		Finally, the top-\textit{k} COPP set \textit{Q} is obtained.
		
		A demonstration of the COPP-Miner algorithm follows.
		
		\begin{example}\label{example10}
			In this example, the time series database \textit{D} shown in Table \ref{tab1} is used. Suppose \textit{minden} = 0.1. The mining process of  top-3 COPPs is as follows.

			Step 1. We extract the extreme points to obtain a new time series database $D’$. 
			
			
			Step 2. We obtain the matching result set consisting of $L_{((1,2), D’)}$ and $L_{((2,1), D’)}$, respectively, and obtain $C_{1,2}$ = $\{$(1,2)$\}$ and $C_{2,2}$ = $\{$(2,1)$\}$.

			Step 3. We select two patterns (1,2) and (2,1) to generate super-patterns. Now, we take (1,2) $\oplus$ (2,1) as an example. It will generate two super-patterns \textbf{u} = (1,3,2) and \textbf{v} = (2,3,1). We use SRC to obtain $r{((1,3,2), D’_+)}$ = 1 and $r((2,3,1), D’_+)$ = 0. According to pruning strategy 1, (2,3,1) is pruned, and (1,3,2) is added into $C_{1,3}$. We calculate $r((1,3,2), D’_-)$ = 0, therefore, $c((1,3,2), D’)$ = 1, and (1,3,2) is added into the top-\textit{k} COPP set \textit{Q}, i.e., \textit{Q} = $\{$(1,3,2)$\}$. After iterating Step 3, we know that \textit{Q} = $\{$(1,3,2), (1,3,2,4), (2,1,4,3)$\}$ since $c((1,3,2,4), D’)$ = 1.0 and $c((2,1,4,3), D’)$ = 0.67.
			

			
			Now, $c_{min} = c((2,1,4,3), D’)$ = 0.67. From Step 5, we obtain $r((2,3,1), D’_-)$ = 0.67. Since $r((2,3,1), D’_-)$ = 0.67 $\leqslant c_{min}$, according to pruning strategy 3, pattern (2,3,1) is also pruned. Similarly, pattern (3,1,2) are pruned, and no new top-\textit{k} COPPs are generated.
			
			Finally, the top-3 COPP set is \textit{Q} = $\{$(1,3,2), (1,3,2,4), (2,1,4,3)$\}$.
		\end{example}
		
		The pseudo-code of COPP-Miner is shown in Algorithm \ref{Algorithm 3}. 
		
		\begin{algorithm}
			\caption{COPP-Miner}
			\label{Algorithm 3}
			\hspace*{0.02in} \leftline{{\bf Input:}
				\textit{D}, \textit{minden}, \textit{k}, and top-\textit{k} COPP set \textit{Q} }
			\hspace*{0.02in} \leftline{{\bf Output:}
				Top-\textit{k} COPP set \textit{Q} }
			\begin{algorithmic}[1]
				\State Use EPE to extract the local extreme points in \textit{D} to obtain a new dataset $D’$;
				\State \textit{Q} $\leftarrow \{\}$;
				\State \textit{Q} $\leftarrow$ ContrastMiner($D’_+$, $D’_-$, \textit{minden}, \textit{k}, \textit{Q});    $//$ Forward mining  
				\State \textit{Q} $\leftarrow$ ContrastMiner($D’_-$, $D’_+$, \textit{minden}, \textit{k}, \textit{Q});    $//$ Reverse mining
				\State \Return \textit{Q};
			\end{algorithmic}
		\end{algorithm}

		Moreover, the pseudo-code of ContrastMiner is shown in Algorithm \ref{Algorithm 4}.
		
		\begin{algorithm}[htb]
			\caption{ContrastMiner}
			\label{Algorithm 4}
			\hspace*{0.02in} \leftline{{\bf Input:}
				$D_1$,\ $D_2$, \textit{minden}, \textit{k}, and top-\textit{k} COPP set \textit{Q}}
			\hspace*{0.02in} \leftline{{\bf Output:}
				Top-\textit{k} COPP set \textit{Q}}
			\begin{algorithmic}[1]
				\State Calculate all occurrences of patterns (1, 2) and (2, 1) in $D’$, and then store patterns into  $C_{1,2}$ and $C_{2,2}$, respectively;
				\State \textit{m} $\leftarrow$ 2;  
				\While {$C_{1,m}$ $<>$ NULL and $C_{2,m}$ $<>$ NULL}
				\State	$C_{1,m+1} \leftarrow$ ContrastPattern($D_1$, $D_2$, \textit{minden}, \textit{k},  \textit{Q}, $C_{1,m}$, $C_{2,m}$);
				\State Sort $C_{1,m+1}$ in descending order according to its support rate value in $D_1$;
				\State	$C_{2,m+1} \leftarrow$ ContrastPattern($D_1$, $D_2$, \textit{minden}, \textit{k},  \textit{Q}, $C_{2,m}$, $C_{1,m}$);
				\State Sort $C_{2,m+1}$ in descending order according to its support rate value in $D_1$;
				\State $m \leftarrow m+1$
				\EndWhile
				\State \Return \textit{Q};
			\end{algorithmic}
		\end{algorithm}
		
		The pseudo-code of ContrastPattern is shown in Algorithm \ref{Algorithm 5}.

		\begin {algorithm}
		\caption{ContrastPattern}\label{Algorithm 5}
		\hspace*{0.02in} \leftline{{\bf Input:}
			$D_1$,\ $D_2$, \textit{minden}, \textit{k}, \textbf{p}, \textbf{q},  top-\textit{k} COPP set \textit{Q}, $C_{1,m}$, and $C_{2,m}$}
		\hspace*{0.02in} \leftline{{\bf Output:}
			$C_{m+1}$ and \textit{Q}}
		\begin{algorithmic}[1]
			\For {each \textbf{p} in $C_{1,m}$}
			\For {each \textbf{q} in $C_{2,m}$}
			\If {\textbf{p} can fuse with \textbf{q}}
			\State Use Algorithm \ref{Algorithm SRC} to obtain $r(\textbf{u}, D_1)$ and its $L_{\textbf{u}, D_1}$, $r(\textbf{v}, D_1)$ and its $L_{\textbf{v}, D_1}$;
			\If {pattern \textbf{u} or \textbf{v} can be pruned by pruning strategies} 
			\State continue; $//$ Prune \textbf{u} or \textbf{v}
			\EndIf
			\State $C_{m+1} \leftarrow C_{m+1} \cup \textbf{u} \cup \textbf{v}$;
			\State 	Use Algorithm \ref{Algorithm SRC} to obtain $r(\textbf{u}, D_2)$ and its $L_{\textbf{u}, D_2}$, $r(\textbf{v}, D_2)$ and its $L_{\textbf{v}, D_2}$;
			\State Calculate $c(\textbf{u}, D’)$ and $c(\textbf{v}, D’)$ according to Definition  \ref{definition5};
			\State If \textbf{u} and \textbf{v} are top-\textit{k} COPPs, then store them in \textit{Q} and update $c_{min}$;
			\EndIf
			\EndFor 
			\EndFor
			\State \Return $C_{m+1}$  and \textit{Q};
		\end{algorithmic}
		\end {algorithm}
		
		\subsection{Complexity analysis }\label{sub4.5}
		
		\begin{theorem} \label{theorm3}
			The space complexity of the COPP-Miner algorithm is $O(T_o + C \times m + T \times m)$, where $T_o$, $C$, $m$, and $T$ are the total length of the original time series, the number of candidate patterns generated, the maximum length of the patterns, and the total length of the time series database after extracting extreme points, respectively.
		\end{theorem}
		\begin{proof}
			The space cost of COPP-Miner consists of three parts: storing the original time series, storing the candidate patterns and storing the occurrences of all patterns. Obviously, the space complexity of storing the original time series is $O(T_o)$. The space complexity of storing the candidate patterns is $O(C \times m)$, since there are $C$ candidate patterns. Now, we analyze the space complexity of storing the occurrences of all patterns. We know that the SRC algorithm calculates the support rate of super-patterns based on the occurrences of sub-patterns. According to Definition 3, an occurrence is represented by the position of a pattern in the time series. Thus, each occurrence corresponds to one position, i.e., a number. Therefore, the space complexity of an occurrence is $O$(1). A position can be many occurrences of patterns with different lengths. For example, in Example 4, $<$4$>$ is an occurrence of pattern (1,3,2), and is also an occurrence of pattern (2,1,4,3). However, a position can at most correspond to an occurrence of a pattern with the same length. Thus, the space complexity of the occurrences of the patterns with the same length is $O(T)$. The maximum length of the patterns is $O(m)$. Therefore, the space complexity of the occurrences of all patterns is $O(T \times m)$. Hence, the space complexity of the COPP-Miner algorithm is $O(T_o + C \times m + T \times m)$.
		\end{proof} 
		
		\begin{theorem} \label{theorm4}
			The time complexity of the COPP-Miner algorithm is $O(T_o + C \times C + T \times m)$.
		\end{theorem}
		\begin{proof}
			The runtime of COPP-Miner also consists of three parts: extracting extreme points, generating all candidate patterns, and calculating the supports of all candidate patterns. Obviously, the time complexity of extracting extreme points is $O(T_o)$, since the length of the original time series is $O(T_o)$, and each value is checked only once. Suppose there are $C$ candidate patterns and the maximum length of the patterns is $m$. Thus, there are $O(C / m)$ candidate patterns for each length on average. The time complexities of determining whether two patterns with length $m$ can be fused and generating a new candidate pattern are both $O(m)$. We propose the group pattern fusion strategy which divides candidate patterns into two groups. Hence, the time complexity of the group pattern fusion strategy for patterns with length $m$ is $O(C/2/m \times C/2/m)$.   Therefore, the time complexity of generating all candidate patterns is $O(m \times m\times C/2/m \times C/2/m)$= $O$($C^{2}$ / 4)= $O$($C^{2}$). As the analysis of the space complexity of the COPP-Miner algorithm, the time complexity of calculating the matching results for all patterns of each length is $O(T)$. The maximum length of all patterns is $m$. Thus, the time complexity of calculating the supports of all candidate patterns is $O(T \times m)$. Hence, the time complexity of the COPP-Miner algorithm is also $O(T_o + C \times C + T \times m)$.
		\end{proof} 				
		
		\section{Experimental results and analysis}\label{section5}
		
		
		To validate the performance of COPP-Miner, we ask the following seven research questions (RQs).

		RQ1: After extracting the local extreme points by EPE, does it have any impact on the performance of COPP-Miner?
		
		RQ2: How does the candidate pattern generation strategy perform, compared with the classic enumeration strategy?
		
		RQ3: Can SRC improve the computational efficiency of COPP-Miner?
		
		RQ4: Can the pruning strategies effectively reduce the number of patterns and improve algorithm performance?

		RQ5: What is the performance of feature extraction of COPP-Miner for classification?
		
		RQ6: Does the value of parameter \textit{minden} affect the running performance and feature extraction for the classification performance?
		
		RQ7: Does the value of parameter \textit{k} affect the running performance and feature extraction for the classification performance?

		To answer RQ1, we propose the COPP-noEPE algorithm and evaluate EPE from two aspects: the running performance is shown in Section \ref{sub5.2} and the classification performance is shown in Section \ref{sub5.3}. In response to RQ2, we develop COPP-enum to explore the effect of the pattern generation strategy in Section \ref{sub5.2}.  To address RQ3, we propose three algorithms, FIM-SW-COPP, KMP-Based-COPP, and Mat-COPP, to investigate the effect of SRC, and the results and analysis are shown in Section \ref{sub5.2}. To answer RQ4, we conducted an ablation experiment in Section \ref{sub-ablation}. To solve RQ5, we select four competitive algorithms: OPP-Miner, OPR-Miner, COPP-noEPE, and COPP-half, and conduct experiments on seven different classifiers, and verify the classification performance by comparing the accuracies of these algorithms in Section \ref{sub5.3}. For RQ6, we verify the running performance with different \textit{minden} values in Section \ref{sub5.4}, and we explore the classification performance with different \textit{minden} values in Section \ref{sub5.5}. For RQ7, we verify the running performance with different \textit{k} values in Section \ref{sub5.6}, and we explore the classification performance with different \textit{k} values in Section \ref{sub5.7}. 
		
		\subsection{Benchmark datasets and baseline methods}\label{sub5.1}
		
		To verify the performance of the COPP-Miner algorithm and its classification performance, we use real time series classification databases as test datasets. All datasets can be downloaded from http://www.timeseriesclassification.com/dataset.php. A specific description of each dataset is given in Table \ref{dataset}.
		
		All experiments were run on a computer with an Intel(R) Core(TM) i5-3230U, 1.60 GHz CPU, 8.0 GB RAM, and Win10 64-bit operating system.
		
		\begin{table}[!htb]
			\renewcommand{\arraystretch}{1.1}	
			\footnotesize
			\captionsetup{font=footnotesize} 
			\caption{Description of benchmark datasets}
			\centering
			\label{dataset}
			\tabcolsep 1pt 
			\begin{tabular}{cccccc}
				\hline\noalign{\smallskip}
				Name & Dataset & $|\textbf{t}|$ & $|D_+|$ & $|D_-|$ &  Total length \\\hline
				DB1 & Mallat-12 & 1024 & 8 & 6 & 14336\\
				DB2 & Wafer & 152 & 97 & 903 & 152000\\
				DB3 & Rock-qm & 2844 & 5 & 5 & 28440\\
				DB4 & Beef & 470 & 24 & 6 & 14100\\
				DB5 & OliveOil-PS & 570 & 13 & 4 & 9690\\
				DB6 & Coffee & 286 & 14 & 14 & 8008\\
				DB7 & Meat & 448 & 20 & 40 & 26880\\
				DB8 & ProximalPhalanxOutlineAgeGroup & 80 & 189 & 211 & 32000\\
				DB9 & DistalPhalanxOutlineAgeGroup & 80 & 143 & 257 & 32000\\
				DB10 & Trace & 275 & 47 & 53 & 27500\\	
				\noalign{\smallskip}\hline
			\end{tabular}
			\begin{tablenotes}
				\item Note1: DB1 is part of the Mallat dataset, which has eight labels. We select the time series with labels 1 and 2 as $D_+$ and $D_-$, respectively.
				\item Note2: DB3 is part of the Rock dataset. There are four labels that correspond to four different types of rocks: mafic, quartzite, marble, and schist. We select the quartzite and marble types of rocks as $D_+$ and $D_-$, respectively.
				\item Note3: Beef dataset contains pure beef and beef with potential adulterations  (heart, tripe, kidney, and liver). We select pure beef and adulterated beef as $D_+$ and $D_-$, respectively.
				\item Note4: DB5 is part of the OliveOil dataset. Each class of this dataset is an extra virgin olive oil from a given country. We select olive oils from Portugal and Spain as $D_+$ and $D_-$, respectively.
				\item Note5: DB7 consists of three classes, chicken, pork, and turkey. We select chicken and turkey as $D_+$ and pork as $D_-$.
				\item Note6: DB8 and DB9 use the outline of one of the phalanges to predict whether the subject is part of one of three age groups: 0–6 years old, 7–12 years old, and 13–19 years old. We select 0–6 and 7–12 years old as $D_+$ and 13–19 years old as $D_-$.
				\item Note7:  DB10 is the Trace dataset which was studied in {\cite{ratanamahatana2004mak}}. For fairness, we use the train samples in {\cite{ratanamahatana2004mak}}, i.e., classes 2 and 6 as $D_+$, and classes 3 and 7 as $D_-$.
				\item Note8: $|\textbf{t}|$ is the length of a single time series.
			\end{tablenotes}
		\end{table}
		
		To verify the performance of COPP-Miner, 13 competitive algorithms are selected. Among them, seven competitive algorithms are proposed to validate the running performance, three competitive algorithms are designed to verify the effectiveness of pruning strategies, and three state-of-the-art algorithms are employed to show the performance of feature extraction for classification.
		

		1) COPP-noEPE and COPP-enum: COPP-noEPE  is proposed to analyze the efficiency of the EPE method and its influence on the classification performance, which does not extract extreme points. COPP-enum is proposed to test the performance of the group pattern fusion strategy, which employs the enumeration strategy to generate super-patterns.

		2) FIM-SW-COPP, KMP-Based-COPP, and Mat-COPP: The three algorithms are proposed to verify the efficiency of the pattern support calculation method in COPP-Miner. FIM-SW-COPP employs the FIM-SW {\cite{wu2022oppm}} method to find COPPs, in which FIM-SW adopts the frequent itemset mining method based on a sliding window to mine frequent OPPs. KMP-Based-COPP uses KMP-Order-Matcher to calculate the support for each candidate pattern, in which KMP-Order-Matcher {\cite{kim2014orde}} is a kind of OPP matching method. Mat-COPP adopts the OPP matching algorithm proposed in {\cite{wu2022oppm}} to calculate the support for each candidate pattern. 
		
		
		3) COPP-original and COPP-minFirst: The two algorithms are used to validate the performance of the support maximum-first strategy.  COPP-original does not apply the support maximum-first strategy, and COPP-minFirst employs the support minimum-first strategy, which is the opposite of the support maximum-first strategy.
		
		4) COPP-noPrun1, COPP-noPrun2, and COPP-noPrun3: The three algorithms are proposed to verify the effectiveness of the pruning strategies, which do not apply Pruning Strategy 1, Pruning Strategy 2,  and Pruning Strategy 3, respectively.
		
		
		5) OPP-Miner, OPR-Miner, and COPP-half: OPP-Miner {\cite{wu2022oppm}} was proposed to mine all frequent OPPs. OPR-Miner {\cite{wu2022oprm}} was used to mine OPRs. We propose COPP-half, which only performs forward mining and does not perform reverse mining.
		

		\subsection{Running performance of COPP-Miner}\label{sub5.2}
		
		To validate the performance of COPP-Miner, we used eight competitive algorithms: COPP-noEPE, COPP-enum, FIM-SW-COPP, KMP-Based-COPP, Mat-COPP, COPP-noPrun, COPP-original, and COPP-minFirst. We performed experiments on the DB1–DB10 datasets, and we set the minimum density threshold \textit{minden} = 0.01 and \textit{k} = 10. Ten rounds of the experiment were performed. Comparisons of the memory consumption, running time, and number of candidate patterns are shown in Tables \ref{tab9m} and \ref{tab9r}, and Fig. \ref{9c}, respectively. Moreover, Table \ref{tab5} shows a comparison between the original dataset length and the dataset length after EPE.
		
		%
		
		\begin{table*}[!htb]
			\renewcommand{\arraystretch}{1.1}	
			\scriptsize
			\captionsetup{font=footnotesize} 
			\caption{Comparison of memory consumption on DB1–DB10 (Mb)}
			\centering
			\label{tab9m}
			\tabcolsep 1pt 
			\begin{tabular*}{\linewidth}{ccccccccccc}
				\hline\noalign{\smallskip}
				Algorithm & DB1 & DB2 & DB3 & DB4 & DB5 & DB6 & DB7 & DB8 & DB9 & DB10\\\hline
				COPP-noEPE & 47.479$\pm$0.110 & 39.551$\pm$0.270 & 26.481$\pm$0.237 & 22.594$\pm$0.279 & 91.459$\pm$0.322 & 88.591$\pm$0.290 & 144.223$\pm$0.232 & 79.605$\pm$0.227 & 75.638$\pm$0.251 & 28.488$\pm$0.266 \\
				COPP-enum & 75.220$\pm$0.131 & 82.553$\pm$0.269 & 59.573$\pm$0.283 & 29.505$\pm$0.279 & 61.547$\pm$0.208 & 146.448$\pm$0.337 & 54.436$\pm$0.268 & 53.400$\pm$0.325 & 38.509$\pm$0.326 & 25.388$\pm$0.294 \\
				FIM-SW-COPP & 104.448$\pm$0.171 & 593.636$\pm$0.248 & 47.501$\pm$0.206 & 114.522$\pm$0.273 & 84.791$\pm$0.108 & 302.352$\pm$0.277 & 205.607$\pm$0.179 & 162.441$\pm$0.230 & 197.505$\pm$0.278 & 510.533$\pm$0.317 \\
				KMP-Based-COPP & 42.623$\pm$0.182 & 141.409$\pm$0.267 & 80.570$\pm$0.190 & 83.622$\pm$0.272 & 129.623$\pm$0.286 & 46.481$\pm$0.303 & 136.609$\pm$0.261 & 76.325$\pm$0.204 & 77.429$\pm$0.288 & 24.600$\pm$0.284 \\
				Mat-COPP & 35.088$\pm$0.121 & 120.437$\pm$0.243 & 49.490$\pm$0.193 & 31.289$\pm$0.108 & 32.561$\pm$0.215 & 34.446$\pm$0.209 & 119.367$\pm$0.177 & 60.511$\pm$0.258 & 63.274$\pm$0.151 & 17.448$\pm$0.272 \\
				COPP-original & 12.137$\pm$0.120 & 29.562$\pm$0.274 & 17.670$\pm$0.244 & 16.460$\pm$0.250 & 14.510$\pm$0.314 & 13.334$\pm$0.210 & 19.486$\pm$0.225 & 30.527$\pm$0.332 & 19.619$\pm$0.255 & 3.818$\pm$0.215\\
				COPP-minFirst & 20.784$\pm$0.134 & 32.630$\pm$0.240 & 20.579$\pm$0.295 & 19.506$\pm$0.207 & 16.556$\pm$0.311 & 14.406$\pm$0.268 & 24.486$\pm$0.134 & 34.400$\pm$0.334 & 20.347$\pm$0.300 & 5.444$\pm$0.255 \\
				COPP-Miner & \textbf{11.447$\pm$0.061} & \textbf{26.418$\pm$0.063} & \textbf{15.877$\pm$0.086} & \textbf{13.223$\pm$0.094} & \textbf{13.488$\pm$0.096} & \textbf{11.303$\pm$0.197} & \textbf{18.839$\pm$0.033} & \textbf{26.537$\pm$0.116} & \textbf{17.397$\pm$0.095} & \textbf{3.599$\pm$0.070} \\
				\noalign{\smallskip}\hline
			\end{tabular*} 
		\end{table*}
		
		\begin{table*}[!htb]
			\scriptsize
			\captionsetup{font=footnotesize} 
			\caption{Comparison of running time on DB1–DB10 (s)}
			\centering
			\label{tab9r}
			\tabcolsep 3pt 
			\begin{tabular*}{\linewidth}{ccccccccccc}
				\hline\noalign{\smallskip}
				Algorithm & DB1 & DB2 & DB3 & DB4 & DB5 & DB6 & DB7 & DB8 & DB9 & DB10\\\hline
				COPP-noEPE & 0.285$\pm$0.007 & 0.417$\pm$0.006 & 0.322$\pm$0.005 & 0.207$\pm$0.008 & 0.191$\pm$0.004 & 0.229$\pm$0.007 & 0.337$\pm$0.013 & 0.364$\pm$0.023 & 0.284$\pm$0.004 & 0.367$\pm$0.008 \\
				COPP-enum & 0.401$\pm$0.007 & 0.441$\pm$0.011 & 0.363$\pm$0.010 & 0.233$\pm$0.009 & 0.855$\pm$0.011 & 0.297$\pm$0.015 & 0.429$\pm$0.021 & 0.551$\pm$0.024 & 0.287$\pm$0.007 & 0.503$\pm$0.012 \\
				FIM-SW-COPP & 1.913$\pm$0.008 & 808.530$\pm$0.366 & 61.455$\pm$0.283 & 83.462$\pm$0.259 & 2.061$\pm$0.030 & 9.328$\pm$0.161 & 6.483$\pm$0.194 & 1.545$\pm$0.236 & 1.850$\pm$0.045 & 661.097$\pm$0.283 \\
				KMP-Based-COPP & 0.193$\pm$0.004 & 0.962$\pm$0.018 & 0.456$\pm$0.016 & 0.394$\pm$0.010 & 0.350$\pm$0.011 & 0.346$\pm$0.009 & 0.308$\pm$0.027 & 0.455$\pm$0.026 & 0.617$\pm$0.013 & 0.853$\pm$0.027\\
				Mat-COPP & 0.174$\pm$0.009 & 0.868$\pm$0.015 & 0.353$\pm$0.009 & 0.349$\pm$0.010 & 0.324$\pm$0.003 & 0.264$\pm$0.005 & 0.273$\pm$0.030 & 0.455$\pm$0.023 & 0.668$\pm$0.009 & 0.714$\pm$0.024\\
				COPP-original & 0.152$\pm$0.007 & 0.221$\pm$0.005 & 0.193$\pm$0.007 & 0.176$\pm$0.003 & 0.185$\pm$0.003 & 0.166$\pm$0.003 & 0.193$\pm$0.003 & 0.253$\pm$0.023 & 0.146$\pm$0.003 & 0.215$\pm$0.003\\
				COPP-minFirst & 0.171$\pm$0.008 & 0.236$\pm$0.004 & 0.205$\pm$0.007 & 0.185$\pm$0.003 & 0.196$\pm$0.003 & 0.186$\pm$0.003 & 0.206$\pm$0.003 & 0.281$\pm$0.008 & 0.155$\pm$0.003 & 0.354$\pm$0.009\\
				COPP-Miner & \textbf{0.140$\pm$0.005} & \textbf{0.192$\pm$0.006} & \textbf{0.177$\pm$0.005} & \textbf{0.170$\pm$0.003} & \textbf{0.172$\pm$0.002} & \textbf{0.141$\pm$0.002} & \textbf{0.178$\pm$0.002} & \textbf{0.215$\pm$0.003} & \textbf{0.124$\pm$0.003} & \textbf{0.202$\pm$0.002}\\
				\noalign{\smallskip}\hline
			\end{tabular*} 
		\end{table*}

		\begin{figure}[!htb]
			\centering
			\includegraphics[width=\linewidth]{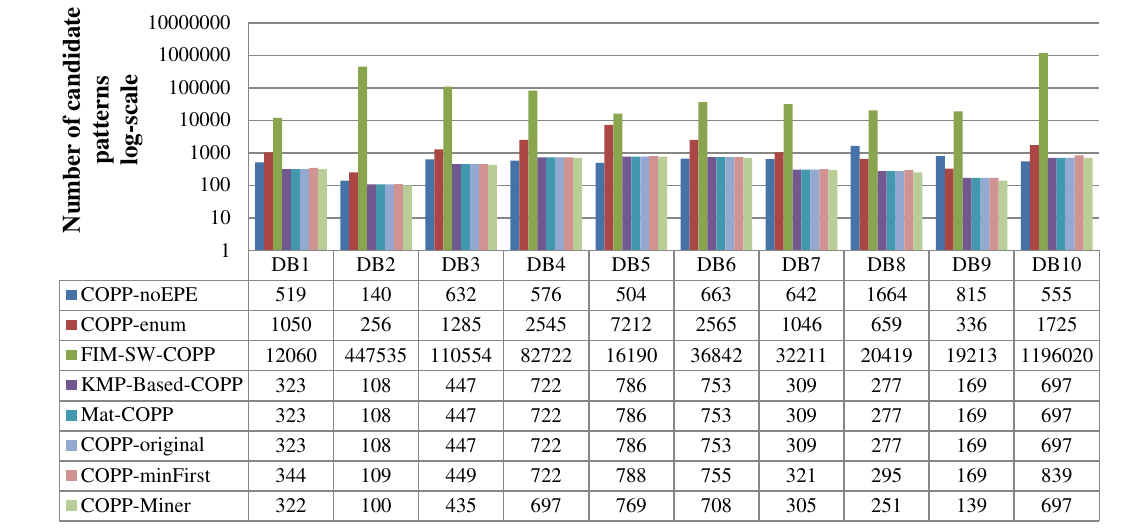}
			\captionsetup{font=footnotesize} 
			\caption{Comparison of number of candidate patterns for DB1–DB10}
			\label{9c}		
		\end{figure}

		\begin{table}[!htb]
			\renewcommand{\arraystretch}{1.1}	
			\scriptsize
			\captionsetup{font=scriptsize} 
			\caption{Comparison between the original dataset length and the dataset length after EPE}
			\centering
			\label{tab5}
			\tabcolsep 1.5pt 
			\begin{tabular}{ccccccccccc}
				\hline\noalign{\smallskip}
				Length&DB1&DB2&DB3&DB4&DB5&DB6&DB7&DB8&DB9&DB10  \\\hline
				\textit{D}&14336&152000&28440&14100&9690&8008&26880&32000&32000&27500  \\
				$D’$&632&62250&4780&2483&823&1555&2172&4948&4777&16006  \\
				\noalign{\smallskip}\hline
			\end{tabular}
		\end{table}
		
		The results give rise to the following observations.
		
		1) COPP-Miner outperforms COPP-noEPE, which indicates that extracting extreme points allows the algorithm to find the classification features more efficiently. For example, Table \ref{tab9m} shows that on DB1, COPP-Miner consumes 11.447 Mb, while COPP-noEPE consumes 47.479 Mb. Table \ref{tab9r} shows that COPP-Miner takes 0.140s, while COPP-noEPE takes 0.285 s. Fig. \ref{9c} shows that COPP-Miner generates 322 candidate patterns, while COPP-noEPE generates 519. We select DB1 as an example. In Table \ref{tab5}, we can see that the original dataset length is 14336, while the dataset length after EPE is 632. We know that the shorter the time series, the less memory the algorithm consumes, the faster the algorithm runs, and the fewer candidate patterns the algorithm generates. Therefore, COPP-Miner occupies less memory, runs faster, and generates fewer candidate patterns. Hence, COPP-Miner outperforms COPP-noEPE.

		
		2) COPP-Miner outperforms COPP-enum, thus demonstrating that the group pattern fusion strategy can efficiently prune candidate patterns. Table \ref{tab9m} shows that COPP-Miner consumes less memory than COPP-enum, and Table \ref{tab9r} shows that COPP-Miner runs faster than COPP-enum.  The reason is that the group pattern fusion strategy can effectively reduce the number of candidate patterns, since the only difference between COPP-Miner and COPP-enum is that the candidate pattern generation strategies are different.  Fig. \ref{9c} shows that  COPP-Miner generates fewer candidate patterns than COPP-enum. The experimental results are therefore consistent with those in Example \ref{example5}. We know that the fewer the candidate patterns, the faster the algorithm runs and the less memory the algorithm consumes. Hence, COPP-Miner outperforms COPP-enum.

		3) The performance of COPP-Miner is better than FIM-SW-COPP, KMP-Based-COPP, and Mat-COPP, since COPP-Miner not only consumes less memory, but also runs faster than FIM-SW-COPP, KMP-Based-COPP, and Mat-COPP.  The reasons are as follows. FIM-SW-COPP uses frequent item sets based on sliding windows to calculate the support, which can be seen as a kind of brute-force method, since the relative order for each window needs to be calculated. More importantly, when the window size is \textit{m}, this method will handle m\textit{}! candidate patterns. Fig. \ref{9c} also reveals this phenomenon, since the number of the candidate patterns of FIM-SW-COPP is significantly greater than other methods. Thus, COPP-Miner is better than FIM-SW-COPP. For KMP-Based-COPP and Mat-COPP, the two algorithms employ pattern matching strategies that cannot use the results of sub-patterns and must scan the database multiple times, which is inefficient. In contrast, COPP-Miner uses the results of the sub-patterns to calculate the support of the super-patterns, avoiding redundant calculations and improving efficiency. Thus, COPP-Miner is better than KMP-Based-COPP and Mat-COPP. Hence, the experimental results indicate that COPP-Miner employing the SRC algorithm is more effective than other competitive algorithms in calculating the support rate.

		4) COPP-Miner is slightly better than COPP-original and COPP-minFirst, which indicates that the support maximum-first strategy can reduce the number of support calculations for candidate patterns and improve the efficiency of COPP-Miner. The difference between COPP-original and COPP-Miner is that COPP-original does not extract extreme points, while COPP-Miner does. COPP-Miner discovers the COPPs on the time series after extracting extreme points which is shorter than the original time series. It is easy to know the shorter the time series, the shorter the runtime, and the less the memory consumption. Hence, COPP-Miner runs faster and consumes less memory than COPP-original. The difference between COPP-minFirst and COPP-Miner is that COPP-minFirst employs the support minimum-first strategy, while COPP-Miner adopts the support maximum-first strategy. Obviously, the higher the support rate in D+, the easier it is to discover contrast patterns with higher contrast rates. According to pruning strategies 2 and 3, the higher the minimum contrast value, the more candidate patterns will be pruned. Therefore, the fewer the support calculations for candidate patterns, the faster the algorithm runs. Hence, COPP-Miner runs faster than COPP-minFirst. All in all, COPP-Miner is slightly better than COPP-original and COPP-minFirst.

		
		
		
		
		\subsection{Ablation experiment}\label{sub-ablation}
		
		
		
		
		
		We designed an ablation experiment to verify the effectiveness of pruning strategies and conducted experiments on DB1–DB10 datasets. In all experiments, we set the minimum density threshold \textit{minden} to 0.01 and \textit{k} to 10. Comparisons of the number of candidate patterns and running time are shown in Figs. \ref{prun5c} and \ref{prun5r}, respectively.

		\begin{figure}[!htb]
			\centering
			\includegraphics[width=\linewidth]{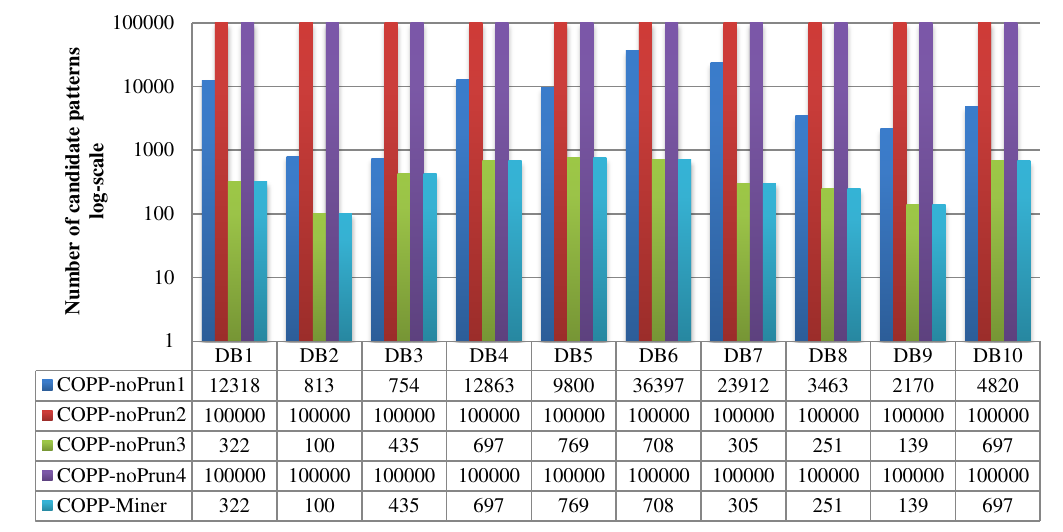}
			\captionsetup{font=footnotesize} 
			\caption{Comparison of number of candidate patterns for ablation experiment}
			\label{prun5c}		
		\end{figure}
		
		\begin{figure}[!htb]
			\centering
			\includegraphics[width=\linewidth]{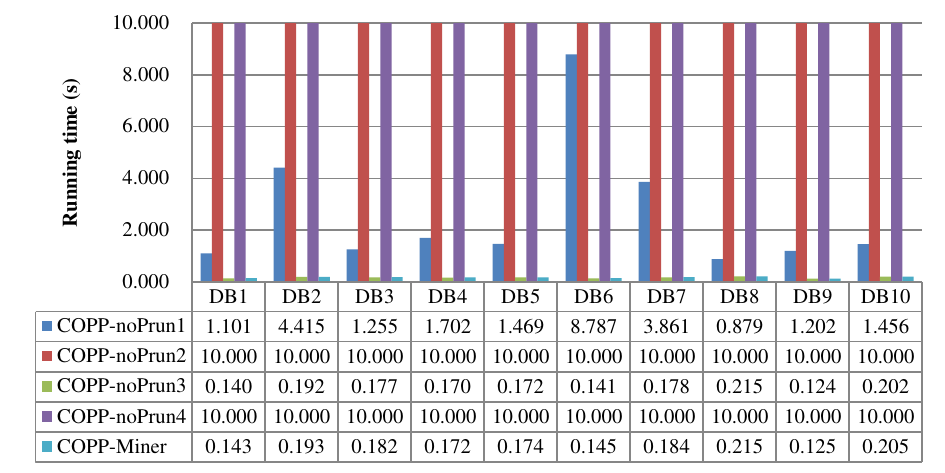}
			\captionsetup{font=footnotesize} 
			\caption{Comparison of running time for ablation experiment}
			\label{prun5r}		
		\end{figure}
		
		The results give rise to the following observations.
		
		1) The experimental results report that COPP-noPrun2 and COPP-noPrun3 algorithms cannot terminate and do not obtain results. Since COPP-noPrun2 and COPP-noPrun3 cannot terminate, the values of candidate patterns and running time of COPP-noPrun2 and COPP-noPrun3 in Figs. \ref{prun5c} and \ref{prun5r} are set to 100000 and 10 s, respectively. Therefore, the results indicate that Pruning Strategies 2 and 3 are crucial, and without them, COPP-Miner cannot stop.
		
		
		2) The number of candidate patterns  and the running time of COPP-Miner are much less than that of COPP-noPrun1, which indicates that Pruning Strategy 1 can effectively reduce the number of candidate patterns, thus improving the efficiency of COPP-Miner. As can be seen from Figs. \ref{prun5c} and \ref{prun5r}, COPP-Miner generates fewer candidate patterns than COPP-noPrun1 and COPP-Miner runs faster than COPP-noPrun1 on all datasets. The reason is as follows. For our problem, users do not know the minimum contrast threshold. Users only need to set parameter \textit{k}. Thus, at the beginning, $c_{min}$ is 0. If we do not employ Pruning Strategy 1, then all patterns must be used as sub-patterns and generate super-patterns, until we find $c_{min}$ $\textgreater$ 0. We know that our mining method in single-class time series satisfies anti-monotonicity which means that if the support of a pattern is 0, then the supports of all its super-patterns are 0. All these patterns cannot be contrast patterns. Thus, if we do not employ Pruning Strategy 1, then many redundant candidate patterns will be generated. Hence, Pruning Strategy 1 is an effective pruning strategy.
		

		\subsection{Classification performance comparison}\label{sub5.3}
		
		To verify the classification ability of COPPs mined by COPP-Miner, we selected four competitive algorithms: OPP-Miner, OPR-Miner, COPP-noEPE, and COPP-half. The mining results of OPP-Miner {\cite{wu2022oppm}}, OPR-Miner {\cite{wu2022oprm}}, COPP-noEPE, COPP-half, and COPP-Miner are denoted as OPP-FP, OPR-FP, COPP-NE, COPP-half, and COPP, respectively. Moreover, the original dataset is denoted as Raw. As we know, this paper focuses on mining the top-\textit{k} COPPs, which refers to the \textit{k} patterns with the highest contrast rates in the time series. To evaluate the performance of the top-10 COPPs, the top-10 patterns are selected as classification features. To avoid the unconvincing classification results of a single classifier, we selected five classical classification algorithms: SVM, C4.5, CART, Adaboost, and KNN, since they are the top-10 algorithms {\cite{wu2008top}}. Then, we selected an adapted version of  KNN, KNN-DTW {\cite{KNN-DTW}}, to work with time series data. Moreover, we also selected two deep learning algorithms, DBC-Forest {\cite{ma2022dbcf}} and HW-Forest {\cite{ma2022hwfo}}. 
		
		The parameters are set as follows. For SVM, we selected the RBF kernel with default parameters. For C4.5 and CART classifiers, we adopted the best features when the decision tree is constructed. In the feature selection process, C4.5 used the maximum information gain ratio criterion, and CART employed the Gini index minimization criterion to select the best features. For Adaboost, n\_estimators is set to 10, and the learning\_rate is set to 1.0. For KNN, parameter \textit{k} is set to 10. For KNN-DTW, parameter \textit{n} is set to 3 and the distance is DTW. For DBC-Forest and HW-Forest, each level was produced by five-fold cross-validation and each forest contained 50 decision trees.
		
		For simplicity, we adopted prediction accuracy as the classification performance measurement and employed a five-fold cross-validation method to evaluate the classification performance of COPPs. The experiments were conducted with \textit{k} = 10 and \textit{minden} = 0.01. Table \ref{tab-accuracy} contains a comparison of the accuracies.

		\begin{table*}[htbp]
			\footnotesize
			\captionsetup{font=footnotesize} 
			\caption{Comparison of accuracy}
			\centering
			\label{tab-accuracy}
			\tabcolsep 4pt 
			\scalebox{0.9}{
			\begin{tabular}{ccccccccccc}
				\hline\noalign{\smallskip}
				Dataset&Feature&SVM&C4.5&CART&Adaboost&KNN& \thead{KNN-\\DTW} & \thead{DBC-\\Forest} & \thead{HW-\\Forest} & Average \\\hline
				\multirow{6}{*}{DB1}
				&Raw& 1.000$\pm$0.000&	0.700$\pm$0.017& 0.783$\pm$0.017&	0.667$\pm$0.003&	0.833$\pm$0.000&	1.000$\pm$0.000&	1.000$\pm$0.000&	1.000$\pm$0.000 & 0.873\\
				&OPP-FP&0.567$\pm$0.000&	0.717$\pm$0.027&	0.783$\pm$0.019&	0.650$\pm$0.000&	0.667$\pm$0.000&  0.786$\pm$0.000&	0.667$\pm$0.002&	0.667$\pm$0.004 &0.688\\
				&OPR-FP&	0.583$\pm$0.000&	0.800$\pm$0.033&	0.800$\pm$0.025&	0.717$\pm$0.033&	0.800$\pm$0.000& 0.857$\pm$0.000&	0.667$\pm$0.003&	0.727$\pm$0.003 & 0.744 \\
				&COPP-NE& 1.000$\pm$0.000&	1.000$\pm$0.000&	1.000$\pm$0.000&	1.000$\pm$0.000&	1.000$\pm$0.000&	1.000$\pm$0.000&	1.000$\pm$0.000&	1.000$\pm$0.000 &1.000\\
				&COPP-half& 1.000$\pm$0.000&	1.000$\pm$0.000&	1.000$\pm$0.000&	1.000$\pm$0.000&	1.000$\pm$0.000& 1.000$\pm$0.000&	1.000$\pm$0.000&	1.000$\pm$0.000 &1.000\\
				&COPP& \textbf{1.000$\pm$0.000}& \textbf{1.000$\pm$0.000}& \textbf{1.000$\pm$0.000}& \textbf{1.000$\pm$0.000}& \textbf{1.000$\pm$0.000}& \textbf{1.000$\pm$0.000}&	\textbf{1.000$\pm$0.000}& \textbf{1.000$\pm$0.000} &\textbf{1.000}\\\hline
				\multirow{6}{*}{DB2}
				&Raw&	0.958$\pm$0.000&	0.974$\pm$0.004&	0.961$\pm$0.002&	0.969$\pm$0.002&	0.992$\pm$0.000& 0.975$\pm$0.000&	0.993$\pm$0.001&	0.993$\pm$0.001 &0.977\\
				&OPP-FP&	0.979$\pm$0.000&	0.997$\pm$0.000&	0.997$\pm$0.001&	0.992$\pm$0.000&	0.998$\pm$0.000&	0.995$\pm$0.000& 0.995$\pm$0.002&	0.995$\pm$0.001 &0.994\\
				&OPR-FP&	0.997$\pm$0.000&	0.997$\pm$0.002&	0.997$\pm$0.003&	0.997$\pm$0.002&	0.998$\pm$0.000& 1.000$\pm$0.000&	0.995$\pm$0.002&	0.995$\pm$0.003 &0.997\\
				&COPP-NE&	1.000$\pm$0.000&	0.998$\pm$0.001&	0.997$\pm$0.001&	0.998$\pm$0.000&	1.000$\pm$0.000& 1.000$\pm$0.000&	1.000$\pm$0.000&	1.000$\pm$0.000 &0.999\\
				&COPP-half&	1.000$\pm$0.000&	1.000$\pm$0.000&	0.998$\pm$0.000&	0.999$\pm$0.000&	1.000$\pm$0.000& 1.000$\pm$0.000&	1.000$\pm$0.000&	1.000$\pm$0.000 &0.999\\
				&COPP&	\textbf{1.000$\pm$0.000} & \textbf{1.000$\pm$0.000} &	\textbf{1.000$\pm$0.000}&	\textbf{1.000$\pm$0.000}&	\textbf{1.000$\pm$0.000}&	\textbf{1.000$\pm$0.000}&	\textbf{1.000$\pm$0.000}&	\textbf{1.000$\pm$0.000} &\textbf{1.000}\\\hline
				\multirow{6}{*}{DB3}
				&Raw&0.667$\pm$0.000&0.500$\pm$0.022&0.555$\pm$0.016&0.544$\pm$0.015&0.667$\pm$0.000&0.888$\pm$0.000&0.500$\pm$0.003&0.666$\pm$0.003 &0.624\\
				&OPP-FP&0.611$\pm$0.000&0.500$\pm$0.008&0.611$\pm$0.053&0.617$\pm$0.000&0.722$\pm$0.000&0.800$\pm$0.000&0.500$\pm$0.002&0.500$\pm$0.003 &0.608\\
				&OPR-FP&0.722$\pm$0.000&0.694$\pm$0.033&0.584$\pm$0.033&0.694$\pm$0.000&0.722$\pm$0.000&0.700$\pm$0.000&0.500$\pm$0.002&0.500$\pm$0.002 &0.639\\
				&COPP-NE&0.583$\pm$0.000&0.806$\pm$0.000&0.806$\pm$0.000&0.806$\pm$0.002&0.806$\pm$0.000&0.900$\pm$0.000&0.750$\pm$0.000&0.750$\pm$0.000 &0.776\\
				&COPP-half&1.000$\pm$0.000&1.000$\pm$0.000&0.917$\pm$0.000&1.000$\pm$0.000&1.000$\pm$0.000&1.000$\pm$0.000&1.000$\pm$0.000&1.000$\pm$0.000 &0.990\\
				&COPP&\textbf{1.000$\pm$0.000}&\textbf{1.000$\pm$0.000}&\textbf{1.000$\pm$0.000}&\textbf{1.000$\pm$0.000}&\textbf{1.000$\pm$0.000}&\textbf{1.000$\pm$0.000}&\textbf{1.000$\pm$0.000}&\textbf{1.000$\pm$0.000} &\textbf{1.000}\\\hline
				\multirow{6}{*}{DB4}
				&Raw&0.933$\pm$0.000 &0.962$\pm$0.003 &0.962$\pm$0.002 &0.962$\pm$0.003 &0.962$\pm$0.000 &0.967$\pm$0.000 &0.958$\pm$0.003 &1.000$\pm$0.000 &0.963\\
				&OPP-FP&0.600$\pm$0.000 &0.833$\pm$0.020 &0.800$\pm$0.005 &0.800$\pm$0.002 &0.867$\pm$0.000 &0.867$\pm$0.000 &0.833$\pm$0.000 &0.867$\pm$0.003 &0.808\\
				&OPR-FP&0.733$\pm$0.000 &0.763$\pm$0.013 &0.767$\pm$0.010 &0.700$\pm$0.008 &0.833$\pm$0.000 &0.867$\pm$0.000 &0.834$\pm$0.001 &0.817$\pm$0.002  &0.790\\
				&COPP-NE&0.667$\pm$0.000 &0.667$\pm$0.000 &0.700$\pm$0.005 &0.667$\pm$0.001 &0.867$\pm$0.000 &0.933$\pm$0.000 &0.867$\pm$0.003 &0.791$\pm$0.000  &0.770\\
				&COPP-half&0.667$\pm$0.000 &0.769$\pm$0.013 &0.833$\pm$0.000 &0.767$\pm$0.000 &0.900$\pm$0.000 &0.900$\pm$0.000 &0.833$\pm$0.000 &0.835$\pm$0.003  &0.814\\
				&COPP&\textbf{0.800$\pm$0.000} &\textbf{0.867$\pm$0.002} &\textbf{0.867$\pm$0.000} &\textbf{0.900$\pm$0.000} &\textbf{0.967$\pm$0.000} &\textbf{0.933$\pm$0.000} &\textbf{1.000$\pm$0.000} &\textbf{1.000$\pm$0.000}  &\textbf{0.916}\\\hline
				\multirow{6}{*}{DB5}
				&Raw&0.756$\pm$0.000&0.702$\pm$0.005&0.759$\pm$0.015&0.880$\pm$0.022&0.833$\pm$0.000&0.875$\pm$0.000&0.750$\pm$0.003&0.750$\pm$0.000 &0.788\\ 
				&OPP-FP&0.599$\pm$0.000&0.544$\pm$0.006&0.644$\pm$0.013&0.611$\pm$0.010&0.767$\pm$0.000&0.824$\pm$0.000&0.700$\pm$0.000&0.705$\pm$0.002 &0.674\\ 
				&OPR-FP&0.522$\pm$0.000&0.711$\pm$0.040&0.589$\pm$0.026&0.656$\pm$0.018&0.767$\pm$0.000&0.765$\pm$0.000&0.706$\pm$0.003&0.750$\pm$0.000 &0.683\\ 
				&COPP-NE&0.578$\pm$0.000&0.822$\pm$0.022&0.633$\pm$0.022&0.761$\pm$0.017&0.778$\pm$0.000&0.941$\pm$0.000&0.750$\pm$0.000&0.750$\pm$0.000 &0.752\\ 
				&COPP-half&0.945$\pm$0.000&0.945$\pm$0.000&0.945$\pm$0.001&0.945$\pm$0.000&0.945$\pm$0.000&0.941$\pm$0.000&1.000$\pm$0.000&1.000$\pm$0.000 &0.958\\ 
				&COPP&\textbf{0.945$\pm$0.000}&\textbf{0.945$\pm$0.000}&\textbf{0.945$\pm$0.003}&\textbf{0.945$\pm$0.000}&\textbf{0.945$\pm$0.000}&\textbf{0.942$\pm$0.000}&\textbf{1.000$\pm$0.000}&\textbf{1.000$\pm$0.000} &\textbf{0.958}\\\hline
				\multirow{6}{*}{DB6}
				&Raw&0.671$\pm$0.000&0.688$\pm$0.002&0.688$\pm$0.002&0.646$\pm$0.003&0.833$\pm$0.000&0.833$\pm$0.000&0.769$\pm$0.000 &0.685$\pm$0.002 &0.726\\ 
				&OPP-FP&0.637$\pm$0.000&0.677$\pm$0.010&0.678$\pm$0.025&0.711$\pm$0.026&0.711$\pm$0.000&0.821$\pm$0.000&0.667$\pm$0.000&0.665$\pm$0.003 &0.696\\ 
				&OPR-FP&0.533$\pm$0.000&0.677$\pm$0.003&0.715$\pm$0.022&0.715$\pm$0.016&0.611$\pm$0.000&0.714$\pm$0.000&0.500$\pm$0.000&0.666$\pm$0.003 &0.641\\ 
				&COPP-NE&0.641$\pm$0.000&0.705$\pm$0.015&0.744$\pm$0.028&0.819$\pm$0.000&0.855$\pm$0.000&0.857$\pm$0.000&0.667$\pm$0.000&0.833$\pm$0.000 &0.765\\ 
				&COPP-half&0.789$\pm$0.000&0.792$\pm$0.015&0.757$\pm$0.015&0.756$\pm$0.012&0.789$\pm$0.000&0.857$\pm$0.000&0.833$\pm$0.000&0.833$\pm$0.000 &0.801\\ 
				&COPP&\textbf{0.822$\pm$0.000}&\textbf{0.759$\pm$0.026}&\textbf{0.757$\pm$0.003}&\textbf{0.793$\pm$0.011}&\textbf{0.889$\pm$0.000}&\textbf{0.893$\pm$0.000}&\textbf{0.833$\pm$0.000}&\textbf{0.833$\pm$0.000} &\textbf{0.822}\\\hline 
				\multirow{6}{*}{DB7}
				&Raw&0.678$\pm$0.000&0.729$\pm$0.003&0.693$\pm$0.002&0.846$\pm$0.000&0.783$\pm$0.000&0.917$\pm$0.000&1.000$\pm$0.000&0.836$\pm$0.003 &0.811\\
				&OPP-FP&0.650$\pm$0.000&0.700$\pm$0.014&0.650$\pm$0.011&0.816$\pm$0.010&0.750$\pm$0.000&0.867$\pm$0.000&0.833$\pm$0.000&0.767$\pm$0.003 &0.754\\
				&OPR-FP&0.650$\pm$0.000&0.800$\pm$0.006&0.769$\pm$0.011&0.817$\pm$0.000&0.733$\pm$0.000&0.850$\pm$0.000&0.750$\pm$0.003&0.786$\pm$0.003 &0.769\\
				&COPP-NE&0.700$\pm$0.000&0.817$\pm$0.005&0.800$\pm$0.010&0.847$\pm$0.006&0.783$\pm$0.000&0.850$\pm$0.000&0.750$\pm$0.000&0.850$\pm$0.004 &0.800\\
				&COPP-half&0.867$\pm$0.000&0.850$\pm$0.007&0.850$\pm$0.007&0.850$\pm$0.000&0.850$\pm$0.000&0.833$\pm$0.000&0.833$\pm$0.000&0.833$\pm$0.000 &0.845\\
				&COPP&\textbf{0.867$\pm$0.000}&\textbf{0.900$\pm$0.006}&\textbf{0.900$\pm$0.007}&\textbf{0.900$\pm$0.000}&\textbf{0.900$\pm$0.000}&\textbf{1.000$\pm$0.000}&\textbf{1.000$\pm$0.000}&\textbf{0.917$\pm$0.000} &\textbf{0.923}\\\hline 
				\multirow{6}{*}{DB8}
				&Raw&0.644$\pm$0.000&0.582$\pm$0.003&0.555$\pm$0.003&0.596$\pm$0.002&0.639$\pm$0.000&0.575$\pm$0.000&0.600$\pm$0.000&0.687$\pm$0.002 &0.610\\
				&OPP-FP&0.828$\pm$0.000&0.857$\pm$0.003&0.878$\pm$0.002&0.867$\pm$0.002&0.920$\pm$0.000&0.912$\pm$0.000&0.914$\pm$0.002&0.875$\pm$0.000 &0.881\\
				&OPR-FP&0.870$\pm$0.000&0.845$\pm$0.004&0.833$\pm$0.004&0.892$\pm$0.001&0.922$\pm$0.000&0.925$\pm$0.000&0.875$\pm$0.000&0.909$\pm$0.002 &0.884\\
				&COPP-NE&0.805$\pm$0.000&0.820$\pm$0.001&0.818$\pm$0.001&0.805$\pm$0.003&0.753$\pm$0.000&0.838$\pm$0.000&0.850$\pm$0.000&0.888$\pm$0.003 &0.822\\
				&COPP-half&0.918$\pm$0.000&0.895$\pm$0.001&0.895$\pm$0.000&0.915$\pm$0.000&0.918$\pm$0.000&0.875$\pm$0.000&0.900$\pm$0.003&0.913$\pm$0.000 &0.904\\
				&COPP&\textbf{0.923$\pm$0.000}&\textbf{0.898$\pm$0.001}&\textbf{0.898$\pm$0.001}&\textbf{0.911$\pm$0.000}&\textbf{0.925$\pm$0.000}&\textbf{0.938$\pm$0.000}&\textbf{0.950$\pm$0.000}&\textbf{0.938$\pm$0.001} &\textbf{0.923}\\\hline
				\multirow{6}{*}{DB9}
				&Raw&0.879$\pm$0.000&0.786$\pm$0.003&0.815$\pm$0.003&0.832$\pm$0.002&0.854$\pm$0.000&0.850$\pm$0.000&0.837$\pm$0.002&0.850$\pm$0.000 &0.838\\
				&OPP-FP&0.882$\pm$0.000&0.817$\pm$0.002&0.812$\pm$0.003&0.853$\pm$0.003&0.782$\pm$0.000&0.905$\pm$0.000&0.813$\pm$0.002&0.850$\pm$0.002 &0.839\\
				&OPR-FP&0.870$\pm$0.000&0.845$\pm$0.000&0.827$\pm$0.003&0.860$\pm$0.002&0.870$\pm$0.000&0.838$\pm$0.000&0.863$\pm$0.002&0.838$\pm$0.003 &0.852\\
				&COPP-NE&0.805$\pm$0.000&0.787$\pm$0.004&0.805$\pm$0.001&0.778$\pm$0.000&0.840$\pm$0.000&0.800$\pm$0.000&0.775$\pm$0.003&0.788$\pm$0.002 &0.797\\
				&COPP-half&0.871$\pm$0.000&0.845$\pm$0.002&0.845$\pm$0.000&0.858$\pm$0.000&0.860$\pm$0.000&0.838$\pm$0.000&0.826$\pm$0.002&0.838$\pm$0.003 &0.847\\
				&COPP&\textbf{0.871$\pm$0.000}&\textbf{0.848$\pm$0.001}&\textbf{0.855$\pm$0.000}&\textbf{0.858$\pm$0.000}&\textbf{0.860$\pm$0.000}&\textbf{0.925$\pm$0.000}&\textbf{0.875$\pm$0.000}&\textbf{0.863$\pm$0.003} &\textbf{0.869}\\\hline
				\multirow{6}{*}{DB10}
				&Raw&0.989$\pm$0.000&0.878$\pm$0.002&0.909$\pm$0.003&0.861$\pm$0.003&0.889$\pm$0.000&1.000$\pm$0.000&0.930$\pm$0.001&0.850$\pm$0.000 &0.913\\
				&OPP-FP&0.990$\pm$0.000&0.860$\pm$0.003&0.844$\pm$0.003&0.886$\pm$0.001&0.950$\pm$0.000&0.990$\pm$0.000&0.900$\pm$0.003&0.850$\pm$0.000 &0.909\\
				&OPR-FP&0.840$\pm$0.000&0.859$\pm$0.002&0.840$\pm$0.000&0.859$\pm$0.000&0.879$\pm$0.000&0.910$\pm$0.000&0.850$\pm$0.000&0.875$\pm$0.000 &0.864\\
				&COPP-NE&0.940$\pm$0.000&0.940$\pm$0.000&0.940$\pm$0.000&0.950$\pm$0.000&0.970$\pm$0.000&0.960$\pm$0.000&0.950$\pm$0.000&0.950$\pm$0.000 &0.950\\
				&COPP-half&0.960$\pm$0.000&0.950$\pm$0.000&0.960$\pm$0.000&0.940$\pm$0.000&0.960$\pm$0.000&0.980$\pm$0.000&1.000$\pm$0.000&0.975$\pm$0.000 &0.966\\
				&COPP&\textbf{0.961$\pm$0.000}&\textbf{0.970$\pm$0.000}&\textbf{0.960$\pm$0.000}&\textbf{0.940$\pm$0.000}&\textbf{0.961$\pm$0.000}&\textbf{1.000$\pm$0.000}&\textbf{1.000$\pm$0.000}&\textbf{1.000$\pm$0.000} &\textbf{0.974}\\ 
				\noalign{\smallskip}\hline
			\end{tabular}
		}
		\end{table*}
		
		The results give rise to the following observations.

		1) COPP can further improve the classification performance of COPP-half. As shown in Table \ref{tab-accuracy}, COPP has a slightly better performance than COPP-half in most cases. The reason is as follows. COPP-half only performs the forward mining which means that COPP-half only discovers the patterns that occur frequently in positive class and infrequently in negative class. However, the patterns that occur frequently in negative class and infrequently in positive class are also contrast patterns and can also be used to classify the time series. Thus, COPP-half misses some top-\textit{k} COPPs, since these missed patterns should be discovered by reverse mining. COPP conducts both forward mining and reverse mining. Hence, COPP can further improve the classification performance of COPP-half.

		
		2) COPP has a better classification performance than COPP-NE. These results indicate that extracting extreme points allows the classification features to be found more efficiently. 
		
		
		3)  Raw, OPP-FP, and OPR-FP have poor classification performances, especially in terms of average accuracy.  The reason is that a COPP is a high contrast rate pattern that occurs frequently in one class and infrequently in the other class. Therefore, it is easy to perform classification using COPPs as classification features regardless of which classifier is used. Moreover, as illustrated in Example \ref{example1}, it is difficult to classify time series using frequent patterns. Therefore, Raw, OPP-FP, and OPR-FP have worse classification performances than COPP. Note that it seems that there are some exceptions. For example, the classification performance of Raw is better than COPP on DB4 dataset, especially in SVM model. However, from Table \ref{dataset}, we know that DB4 is an unbalanced dataset. For an unbalanced dataset, F1-score is the commonly used indicator. Thus, the comparison of F1-score on DB4 is shown in Table \ref{tab-F1}. As can be seen from Table \ref{tab-F1}, COPP-Miner is also better than other methods.
		

		\begin{table}[htbp]
			\renewcommand{\arraystretch}{1.1}
			\scriptsize
			\captionsetup{font=scriptsize} 
			\caption{Comparison of F1-score on DB4}
			\centering
			\label{tab-F1}
			\tabcolsep 1pt 
			\begin{tabular}{cccccccccc}
				\hline\noalign{\smallskip}
				Feature& SVM& C4.5& CART& Adaboost& KNN&  \thead{KNN-\\DTW} & \thead{DBC-\\Forest} & \thead{HW-\\Forest} & Average \\\hline
				Raw&0.759 &0.845 &0.881 &0.798 &0.830 &0.817 &0.917 &1.000 &0.856 \\
				OPP-FP&0.601 &0.723 &0.646 &0.759 &0.759 &0.817 &0.869 &0.727 &0.738 \\
				OPR-FP&0.794 &0.715 &0.756 &0.715 &0.759 &0.679 &0.800 &0.667 &0.736 \\
				COPP-NE&0.750 &0.869 &0.906 &0.729 &0.817 &0.817 &0.923 &0.833 &0.831 \\
				COPP-half&0.779 &0.737 &0.741 &0.694 &0.798 &0.900 &0.880 &0.875 &0.801 \\
				COPP&\textbf{0.845} &\textbf{0.800} &\textbf{0.808} &\textbf{0.833} &\textbf{0.884} &\textbf{0.917} &\textbf{1.000} &\textbf{1.000} &\textbf{0.886} \\
				\noalign{\smallskip}\hline
			\end{tabular}
		\end{table}
		
		
		\subsection{Influences of different parameter \textit{minden}}\label{sub5.4}
		
		In this section, we report the influences of different parameter \textit{minden} on classification performance. For clarity, we set \textit{k} = 10. The memory consumption, running time, and number of candidate patterns with different $minden$ values are shown in Figs. \ref{den6m}, \ref{den6r}, and	 \ref{den6c}, respectively.
		
		\begin{figure}[!htb]
			\centering
			\includegraphics[width=\linewidth]{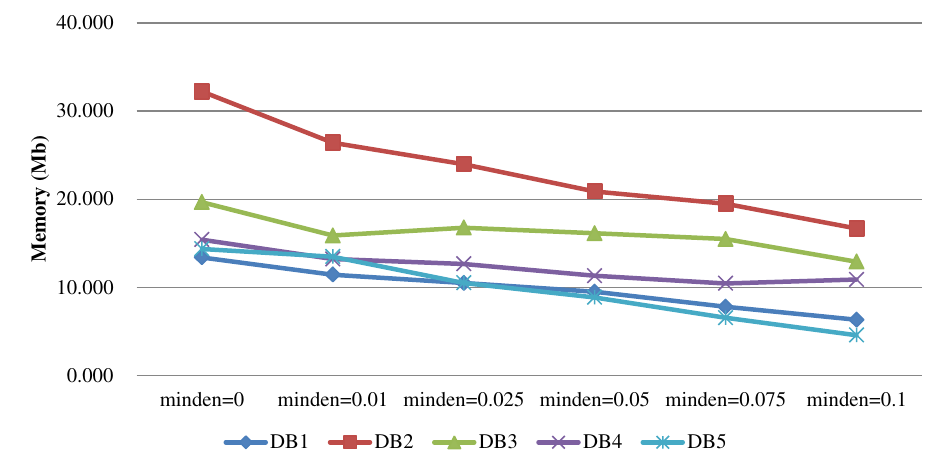}
			\captionsetup{font=footnotesize} 
			\caption{Memory consumption with different \textit{minden} values on DB1–DB5}
			\label{den6m}		
		\end{figure}
		
		\begin{figure}[!htb]
			\centering
			\includegraphics[width=\linewidth]{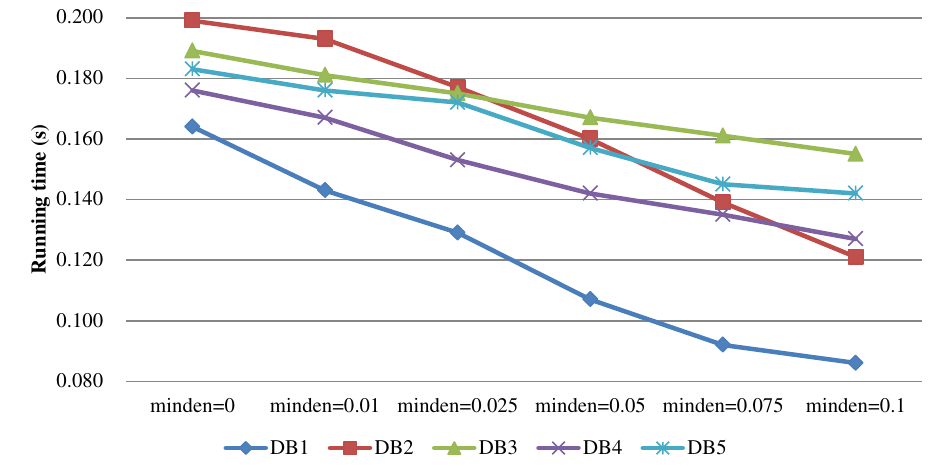}
			\captionsetup{font=footnotesize} 
			\caption{Running time with different \textit{minden} values on DB1–DB5}
			\label{den6r}		
		\end{figure}
		
		\begin{figure}[!htb]
			\centering
			\includegraphics[width=\linewidth]{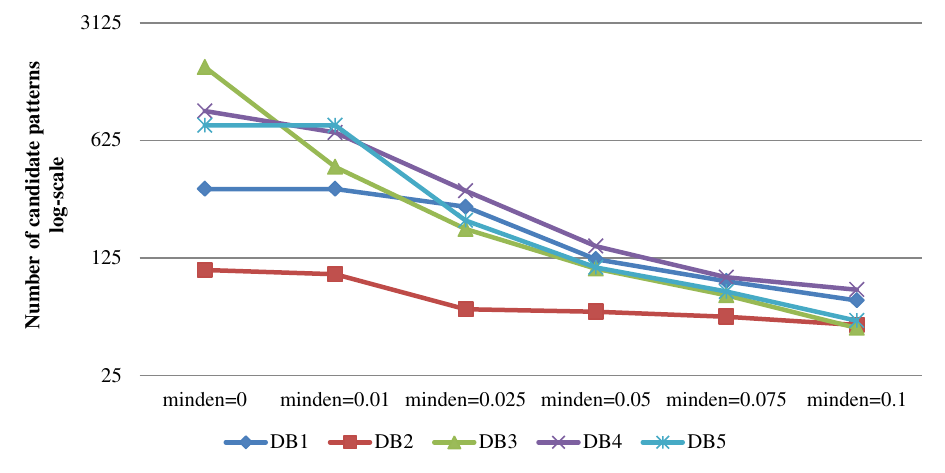}
			\captionsetup{font=footnotesize} 
			\caption{Number of candidate patterns with different \textit{minden} values on DB1–DB5}
			\label{den6c}		
		\end{figure}
		
		
		In Figs. \ref{den6m} to \ref{den6c}, the fluctuation trend of the line chart is downward, that is, in general, as the \textit{minden} value increases, the consumption of memory decreases and the running speed increases, and fewer candidate patterns are generated.  The reasons are as follows. First, with the increase of the \textit{minden} value, the number of candidate patterns whose density is greater than the threshold \textit{minden} decreases significantly, which makes the algorithm consume less memory and run faster. In addition, with the increase of the \textit{minden} value, the minimal contrast rate decreases. That is, more candidate patterns will be generated. Therefore, in general, with the increase of \textit{minden}, the memory consumption decreases, the running speed increases, and the number of candidate patterns decreases.

		\subsection{Classification performance of different parameter \textit{minden}}\label{sub5.5}
		
		To verify the classification performance of COPPs with different \textit{minden} values, six different thresholds are 0, 0.01, 0.025, 0.05, 0.075, and 0.1, and KNN is adopted as the classifier. Fig. \ref{den6acc} shows the classification accuracy with different \textit{minden} values.
		
		\begin{figure}[!htb]
			\centering
			\includegraphics[width=\linewidth]{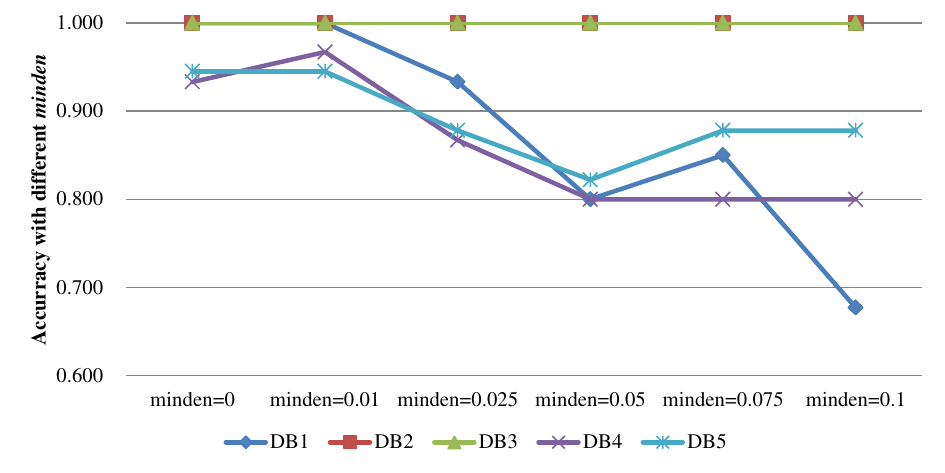}
			\captionsetup{font=footnotesize} 
			\caption{Classification results with different \textit{minden} values on DB1–DB5}
			\label{den6acc}		
		\end{figure}
		
		
		As shown in Fig. \ref{den6acc}, the parameter \textit{minden} has a certain influence on the classification accuracy. The overall trend is that the larger the \textit{minden} value, the lower the classification accuracy. The reason is that with the increase of the \textit{minden} value, the minimal contrast rate decreases.
		
		\subsection{Influences of different parameter \textit{k}}\label{sub5.6}
		
		In this section, we report the influences of different parameter \textit{k}, and we set \textit{minden} = 0.01. The number of candidate patterns, memory consumption, and running time with different \textit{k} values are shown in Figs. \ref{k5c}, \ref{k5m}, and	 \ref{k5r}, respectively.

		\begin{figure}[!htb]
			\centering
			\includegraphics[width=\linewidth]{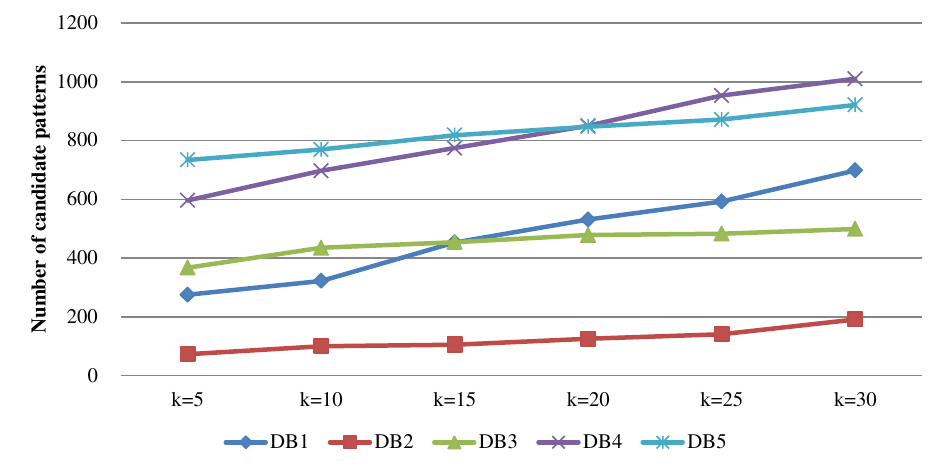}
			\captionsetup{font=footnotesize} 
			\caption{Number of candidate patterns with different \textit{k} values on DB1–DB5}
			\label{k5c}		
		\end{figure}
		
		\begin{figure}[!htb]
			\centering
			\includegraphics[width=\linewidth]{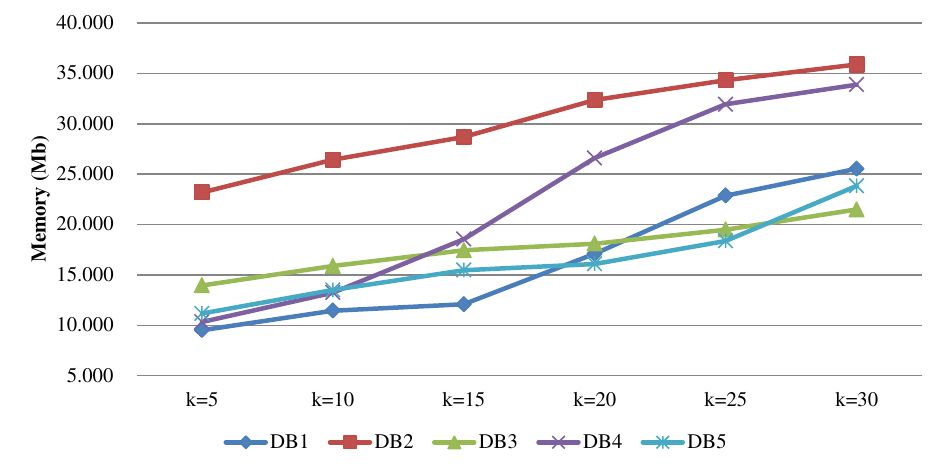}
			\captionsetup{font=footnotesize} 
			\caption{Memory consumption with different \textit{k} values on DB1–DB5}
			\label{k5m}		
		\end{figure}
		
		\begin{figure}[!htb]
			\centering
			\includegraphics[width=\linewidth]{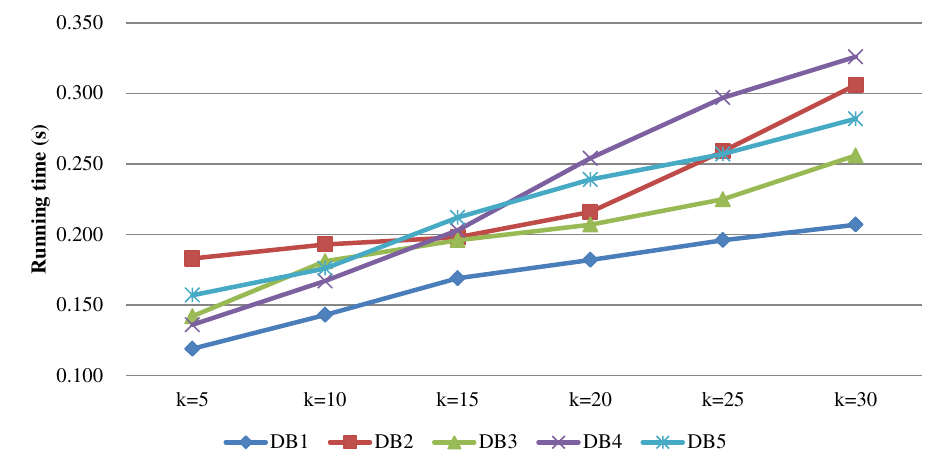}
			\captionsetup{font=footnotesize} 
			\caption{Running time with different \textit{k} values on DB1–DB5}
			\label{k5r}		
		\end{figure}
		
		
		From Figs. \ref{k5c} to \ref{k5r}, we know that the larger the \textit{k} value, the more the candidate patterns, the greater the memory consumption, and the shorter the runtime. The reason is as follows. The increase of the \textit{k} value means a decrease in the minimal contrast rate and an increase in the number of mined patterns. Thus, more candidate patterns will be generated. According to Theorems 6 and 7, we know that the space and time complexities of COPP-Miner are positively correlated with the number of candidate patterns. Thus, with the increase of candidate patterns, the algorithm consumes more memory and more running time. Therefore, in general, with the increase of \textit{k}, the number of candidate patterns increases, the memory consumption increases, and the running time decreases.

		\subsection{Classification performance of different parameter \textit{k}}\label{sub5.7}
		
		To verify the classification performance of COPPs with different \textit{k} values, six different \textit{k} are 5, 10, 15, 20, 25 and 30, and we set \textit{minden} = 0.01. Moreover, we select KNN as the classifier and the parameter of KNN is five. Fig. \ref{k5acc} shows the classification accuracy with different \textit{k} values.
		
		\begin{figure}[!htb]
			\centering
			\includegraphics[width=\linewidth]{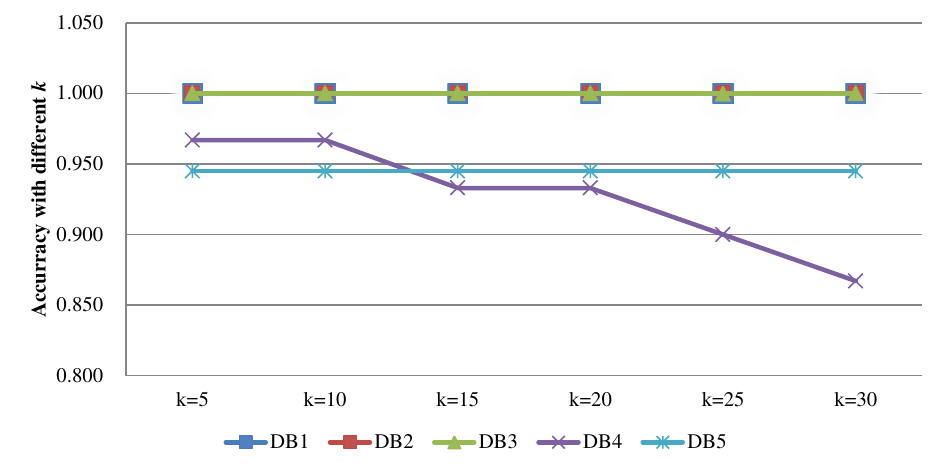}
			\captionsetup{font=footnotesize} 
			\caption{Classification results with different \textit{k} values on DB1–DB5}
			\label{k5acc}		
		\end{figure}

		
		As shown in Fig. \ref{k5acc}, parameter \textit{k} has no effect on the classification accuracy in all datasets, except for DB4. On DB4, the overall trend is that the larger the \textit{k} value, the lower the classification accuracy.  The reason is that with the increase of the \textit{k} value, the minimal contrast rate decreases. Thus, the number of COPPs with low contrast is increased, and contrast is used to represent the significant difference between the two categories of datasets. Therefore, the lower the contrast of patterns, the worse the classification performance. Hence, with the increase of \textit{k} value, the classification accuracy of the algorithm decreases.
		
		\section{Conclusion and future work}\label{section6}
		
		To extract features for time series classification, we utilize the top-\textit{k} COPP mining, where a contrast pattern is a pattern that occurs frequently in one class and infrequently in the other class, and OPP refers to the relative order of a sub-time series, which can easily represent a trend. To effectively discover the top-\textit{k} COPPs, we propose the COPP-Miner algorithm, which consists of three parts: 1) EPE, 2) forward mining, and 3) reverse mining. To compress data, we employ EPE which can effectively obtain the key trends of a time series by extracting extreme points. To effectively mine COPPs, we use forward mining to find COPPs, which contains three steps: the group pattern fusion strategy to generate candidate patterns, the SRC method to efficiently calculate the pattern support, and the pruning strategies to further prune candidate patterns. Finally, reverse mining uses one pruning strategy to prune candidate patterns and consists of applying the same process as forward mining but positive and negative sequences are swapped. The extensive experimental results indicate that COPP-Miner not only has a better running performance than other competitive algorithms, but also discovered top-\textit{k} COPPs can be used as features to improve the performance of time series classification.
		
		
		The following work can be considered in the future.
		
		
		1. For balanced binary classification problems, this paper proposes COPP-Miner that mines contrast order-preserving patterns to extract features for time series classification. How to mine the contrast order-preserving patterns for multi-class classification problems or imbalanced binary classification problems is worth further exploration. 
		
		2. In addition, COPP-Miner is aimed at mining one-dimensional time series. How to mine order-preserving patterns for multi-dimensional data is also worth investigating.
		
		3. Moreover, through Sections \ref{sub5.4} and \ref{sub5.5}, we have verified that the value of \textit{minden} has an impact on the classification performance, and the smaller the \textit{minden}, the lower the speed, but the better the classification performance. The value of \textit{minden} ranges from 0 to 1, and the minimum value is 0. For users without prior knowledge, they can directly select 0 to obtain the optimal classification performance. How to achieve the optimal classification performance and simultaneously reduce the running time which is a multi-objective optimization problem is worth investigating.

\section*{Acknowledgement}
This work was partly supported by the National Natural Science Foundation of China (62372154, 61976240, 52077056, 62120106008).

\end{document}